\documentclass{article}

\usepackage{graphicx} % Required for inserting images
\usepackage[utf8]{inputenc}
\usepackage[affil-it]{authblk}
\usepackage{graphicx}
\usepackage{subcaption}
\usepackage{comment,amsmath,amsbsy,amssymb,graphicx,dsfont,upgreek,textcomp,braket,setspace,blindtext,hyperref,verbatim,amsthm,mathrsfs,mathtools,graphicx,float,enumerate,cleveref,cases,bigints}
\usepackage[font={scriptsize,it}]{caption}
\usepackage[margin=1in]{geometry}
\hypersetup{
    colorlinks=true, %set true if you want colored links
    linktoc=all,     %set to all if you want both sections and subsections linked
    linkcolor=blue,  %choose some color if you want links to stand out
}
\usepackage{tocloft} 
\usepackage[table,svgnames]{xcolor}
\usepackage{tikz}
\usetikzlibrary{arrows}
\usepackage[toc,page]{appendix}
\usepackage{xcolor}

\Crefname{appsec}{appendix}{appendices}
\numberwithin{equation}{section}

\newtheorem{lemma}{Lemma}[section]
\newtheorem{corollary}{Corollary}[section]
\newtheorem{remark}{Remark}[section]
\newtheorem{definition}{Definition}[section]
\newtheorem{proposition}{Proposition}[section]

\title{\Large Nonlinear excitations in multi-dimensional nonlocal lattices\vspace{-1ex}}
\author{Brian Choi \vspace{-2ex}}
 
\author{Brian Choi\thanks{Corresponding author. University of Tennessee at Chattanooga, \texttt{choigh@bu.edu}\\
\textbf{Keywords:} nonlocal dynamics, excitation threshold, localization, breather, dispersive decay}
\vspace{-2ex}  }

\date{}
\begin{document}
\maketitle\vspace{-4ex}
\maketitle
\begin{abstract}
We study the formation of breathers in multi-dimensional lattices with long-range interactions. By variational methods, the exact relationship between various parameters (dimension, nonlinearity, nonlocal parameter $\alpha$) that defines positive excitation thresholds is characterized. We establish a sharp mass-threshold dichotomy: no positive threshold in the mass-subcritical regime, and a strictly positive threshold at and above the critical regime. In the anti-continuum regime, a family of unique ground states characterizes the excitation thresholds, enabling explicit computations. Analytic formulas of the excitation thresholds, determined by the ground states, are derived and corroborated with numerical simulations. We not only characterize the sharp spatial decay of ground states, which varies continuously in $\alpha$, but also identify the time decay of dispersive waves, which undergoes a discontinuous transition in $\alpha$.
\end{abstract}
\noindent \textbf{MSC 2020} 34A08, 34A12, 37K58, 37K60, 37K40

%\tableofcontents

%\classification{35B30, 35Q40, 35Q55, 35Q60, 35R11, 37K60}\\
%\keywords{Continuum Limit, Fractional Equation, Lattice system, NLS}
%\begin{spacing}{0.001}
%\tableofcontents
%\end{spacing}
%\singlespacing
\section{Introduction}
Discrete breathers on nonlinear lattices, solitary waves that are localized in space and periodic in time, have been an active area of research in both experiments and mathematical studies for decades \cite{kartashov2011solitons,kevrekidis2011non}. Our focus is on discrete Hamiltonian systems with long-range interactions (LRI) whose equations of motion provide a variety of dynamics ranging from breathers with slowly decaying tails to nonlinear dispersive decay due to an interplay between linear dispersion (or lattice diffraction) characterized by the nonlocal lattice coupling and nonlinearity. LRI arises naturally in the numerical discretization of fractional differential equations, among many others. Relevant to physical applications and numerical studies are the following non-exhaustive references: \cite{mingaleev1999models} (energy transport in biomolecules), \cite{efremidis2002discrete,szameit2008long} (optical waveguide arrays), \cite{liu2023experimental} (nonlinear optical realization of the fractional Schr\"odinger equation), \cite{gaididei1997effects,choudhury1996long,korabel2007transition} (quantum systems), \cite{kevrekidis2009non,kevrekidis2013drastic,christiansen1998solitary} (existence and stability of breathers), and \cite{flach1998breathers,moleron2019nonlinear} (breathers in the Klein-Gordon-type coupled oscillators with global coupling). Among many sub-topics related to breathers and LRI, our interest is in extending the theory of excitation thresholds beyond the nearest-neighbor interaction (NNI) \cite{weinstein1999excitation,flach1997energy} to LRI with a nonlocal generalization of the discrete nonlinear Schr\"odinger equation (DNLS) as the main model; see also \cite{kastner2004energy,jcuevas2008thresholds,cuevas2010energy,karachalios2013breathers,cheng2015solitary,cheng2016gap} for excitation thresholds of DNLS featuring nonlinear hopping, saturable nonlinearities, and periodic/gap solitons.

To illustrate the existence of excitation thresholds (see \eqref{def:excite}, \eqref{excitation_threshold} for a precise definition), consider the focusing DNLS
\begin{equation}\label{dnls}
    i \dot{u}_n = -\Delta_d u_n - |u_n|^{p-1} u_n
\end{equation}
where $u:\mathbb{Z}^d \times \mathbb{R} \rightarrow \mathbb{C},\ p > 1$ and $\Delta_{d} u_n := -2d u_n + \sum \limits_{|m-n|=1} u_m$ with $|\cdot|$ being the $l^1$ norm on $\mathbb{Z}^d$. Comparing \eqref{dnls} with the NLS $i \partial_t u = -\Delta u - |u|^{p-1} u$, both \eqref{dnls} and NLS are invariant under the time-translation and gauge symmetry, which yield energy and mass conservation, respectively. Although the solutions of NLS are invariant under $u(x,t) \mapsto \lambda^{-\frac{2}{p-1}} u(\frac{x}{\lambda},\frac{t}{\lambda^2})$ for $\lambda > 0$, there is no such scaling symmetry for \eqref{dnls} due to the underlying lattice structure.

For $\omega > 0,\ Q(x) \in \mathbb{R}$, substituting $u(x,t) = e^{i \omega t}Q(x)$ into NLS, we obtain a nonlinear eigenvalue equation $\omega Q = \Delta Q + |Q|^{p-1} Q$. Since $\| \lambda^{-\frac{2}{p-1}} Q(\frac{x}{\lambda}) \|_{L^2} = \lambda^{\frac{d}{2} - \frac{2}{p-1}} \| Q \|_{L^2}$, if $Q$ is a non-trivial solution for some $\omega > 0$ and $p \neq 1 + \frac{4}{d}$, another non-trivial solution to the eigenvalue equation with $\omega^\prime>0$ with an arbitrarily small $L^2$-norm can be constructed by the scaling symmetry. In other words, the NLS allows a family of breather solutions at any $L^2$-norm (zero excitation threshold, or in short, $\nu_0 = 0$) unless $p = 1 + \frac{4}{d}$; if $p = 1 + \frac{4}{d}$, then any initial data with $L^2$-norm less than $\| Q \|_{L^2}$, where $Q \in H^1(\mathbb{R}^d)$ is the ground state of NLS, disperse and decay to zero in $L^\infty(\mathbb{R}^d)$ as $t \rightarrow \infty$ \cite{cmp/1103922134} and therefore some positive $L^2$-norm ($\nu_0 > 0$) is required for localization to persist globally in time.

With the lack of scaling symmetry, the dynamics of lattice systems exhibits distinct behavior relative to their continuum analogue. By variational methods, the existence of excitation thresholds for DNLS was shown and classified in \cite{weinstein1999excitation}, providing a rigorous proof of the conjecture originally formulated in \cite{flach1997energy}. More precisely, it was shown that $\nu_0 > 0$ if and only if $p \geq 1 + \frac{4}{d}$. Motivated by the work above, our goal is to investigate the role of LRI in excitation thresholds (\Cref{p2,p1}). For the coupling strength decaying algebraically as $|n|^{-(d+\alpha)}$, we show that any power-type nonlinearity is mass-supercritical for $0 < \alpha \ll 1$, and hence $\nu_0 > 0$. Our definition of LRI requires the coupling coefficients to be summable, and therefore $|n|^{-d}$ is not considered. However, \Cref{p4} considers this singular case under a logarithmic weight where it is shown that $\nu_0 > 0$ for any $p>1$. On the other hand, when $\alpha \geq 2$, the effect of LRI is replaced by that of NNI, which recovers the result of \cite{weinstein1999excitation}.

Once the criteria of existence/non-existence of ground states have been identified, it is of interest to study the spatial profile of ground states, their stability properties, and a generic dynamical behavior of initial data near ground states as $t \rightarrow \infty$. We revisit the topic of the slowly-decaying tail behavior of ground states \cite{gaididei1997effects,flach1998breathers,moleron2019nonlinear} induced by LRI and show the fully algebraic nature of the decay of \eqref{main_model} via rigorous error analysis including the sharp order of decay. At the anti-continuum regime (\Cref{onsite}), the ground states are unique, modulo symmetries, and are minimizers of the Gagliardo-Nirenberg type inequalities \eqref{gns5} defined similarly to the Weinstein functional \cite{cmp/1103922134}. This determines an explicit formula for the excitation thresholds. In \Cref{applications}, we numerically confirm the excitation of breather-like solutions for localized initial data with sufficiently large mass. The solutions decompose into a hump of peak intensity and radiation where the radiation component, for any time, decays algebraically in space, consistent with \Cref{algebraic_decay}. For sufficiently small mass, the solutions exhibit global scattering and nonlinear dispersive decay. It is shown, by a detailed study on the polylogarithms (see also \cite[Appendix]{gori2013modulational}), that $\alpha = 1$ is the threshold nonlocal parameter at which the order of dispersive decay exhibits a discontinuous transition (see \Cref{p5}); interestingly, this dichotomy persists in the continuum limit as suggested in \cite[Remark 2]{hong2019strong}. In general, the non-analyticity of polylogarithms at the origin affects the dispersive properties of \eqref{main_model}, and thus has consequences on the global dynamical behavior.

An outline of this paper is as follows: basic definitions are given in \Cref{formulation}; \Cref{optimization} revisits the method of concentration compactness in obtaining breather solutions and confirms its applicability to LRI; the sharp decay of ground states is given in \Cref{decay3}; the precise relationship between $d,p,\alpha$ and the excitation threshold is given in \Cref{ET}; that the ground states near the anti-continuum regime are unique minimizers of the excitation threshold functional is given in \Cref{ET2}; in \Cref{applications}, the initial value problem in 1D lattice is given as an application of previous ideas.
\subsection{Formulation}\label{formulation}

The behavior of breathers depends crucially on the nonlocal coupling $J = (J_n)$ and its associated tail index $\alpha \in (0,\infty]$. For $k \in \mathbb{N}$, let $\mu_k$ be the probability measure on $S_k = \{n \in \mathbb{Z}^d: |n|=k\}$ given by $\mu_k \{n\} = \frac{J_n}{\sum\limits_{|m|=k} J_m}$. Let $\mathbb{S}^{d-1} \subseteq \mathbb{R}^d$ be the standard unit sphere under the Euclidean norm $|\cdot|_{2}$.

\begin{definition}\label{def:kernel}

Let $J:\mathbb{Z}^d \rightarrow [0,\infty)$ satisfy $J(0,\dots,0) = 0$ (no self-interaction), $J_{n} = J_{-n}$ for all $n \in \mathbb{Z}^d$ (symmetry), and $\overline{J} := \sum\limits_{m \in \mathbb{Z}^d} J_m < \infty$ (summability). Let $\eta:\mathbb{N} \rightarrow [0,\infty)$ be an envelope such that (i) if $\alpha \in (0,\infty)$, assume there exist $0< C_1 < C_2$ and $k_0 \in \mathbb{N}$ such that $C_1 \leq k^{d+\alpha} \eta(k) \leq C_2$ for all $k \geq k_{0}$; (ii) if $\alpha = \infty$, assume $\lim\limits_{k \rightarrow \infty} k^{d+\beta} \eta(k) = 0$ for any $\beta < \infty$. Further assume 
\begin{enumerate}
    \item[(I)] \textbf{Stable-like in the far-field:} there exist $R, c_{\pm} > 0$ such that for any $k \geq R$ and $n \in \mathbb{Z}^d$ with $|n| \geq R$, 
    \begin{equation*}
    J_n \leq c_{+} \eta(|n|),\qquad \overline{J}_k := \sum_{m \in S_k} J_m \geq c_{-} k^{d-1} \eta(k).    
    \end{equation*} 
    
    \item[(II)] \textbf{Nearest-neighbor nondegeneracy:}
    \begin{equation*}
        \min_{m \in S_1} J_m > 0.
    \end{equation*}    
\end{enumerate}
\end{definition}

Define
\begin{equation}\label{operator}        Lu_n = \sum_{m \in \mathbb{Z}^d \setminus{\{n\}}} J_{n-m} (u_n - u_m),
\end{equation}
from which our main model is given by
\begin{equation}\label{main_model}
    i \dot{u}_n = \kappa L u_n - |u_n|^{p-1} u_n;\ -\omega q_n = \kappa L q_n - |q_n|^{p-1} q_n,
\end{equation}
where the time-independent model is obtained by taking $u_n = e^{i \omega t} q_n$ for $\omega > 0$; unless the anti-continuum limit is discussed, assume $\kappa = 1$. The two conserved quantities
\begin{equation*}
M[u] := \| u \|_{l^2}^2 = \sum_{n} |u_n|^2;\ H[u] := \frac{1}{2} \langle Lu, u\rangle -\frac{\| u \|_{l^{p+1}}^{p+1}}{p+1}
\end{equation*}
are used throughout. By direct estimation, $L$ is bounded on $l^2(\mathbb{Z}^d)$ with an upper bound in operator norm $2\overline{J}$. Define the discrete Fourier transform $\widehat{u}(\xi) = \mathcal{F}[u](\xi) = \sum\limits_{m \in \mathbb{Z}^d} u_m e^{-i m \cdot \xi}$ for $\xi \in \mathbb{T}^d = (-\pi,\pi]^d$, and let $|\xi|$ denote the Euclidean norm on the Fourier domain. Then, the symbol of $L$ 
\begin{equation*}
\sigma(\xi) = 2\sum\limits_{m \neq 0} J_{m} \sin^2 \left(\frac{m \cdot \xi}{2}\right)    
\end{equation*}
is given as a real-valued function, not complex-valued, due to the even symmetry of $J$. The nearest-neighbor nondegeneracy implies that $\sigma(\xi) = 0$ if and only if $\xi = 0$ on $\mathbb{T}^d$, and therefore \eqref{operator} defines a positive-definite operator. The behavior of $\sigma(\xi)$ at the origin (see \Cref{l3}) plays a key role in determining excitation thresholds. We show that the symbol is bounded above and below by that of the (fractional) continuum Laplacian, thereby extending \cite[Appendix A.1]{kirkpatrick2013continuum} into a higher-dimensional setting. The limit condition $k^{1+\alpha} J_{k} \xrightarrow[k\rightarrow \infty]{} C_{J} \in (0,\infty)$ was imposed in \cite{kirkpatrick2013continuum} to study continuum limit. We relax the tail behavior from the limit to bounds, since our model is inherently discrete.

Examples of $J$ with $\alpha = \infty$ include the graph Laplacian, when $J_n = 1$ for $|n|=1$ and $J_n = 0$ otherwise, and any coupling that decays exponentially independent of lattice directions. For $\alpha < \infty$, our assumption on $\eta$ generalizes the radial kernel $J_{n} = |n|^{-(d+\alpha)}$. Although the summability fails for $|n|^{-d}$ for $\alpha = 0$, excitation thresholds are obtained under logarithmic refinements in \Cref{p4}. In \Cref{non-radial}, a specific example of non-radial coupling is considered. A similar phenomenon could potentially arise from direction-dependent LRI, which could have promising applications in photonics and optoelectronics \cite{li2021review}. For mathematical perspectives, see \cite{gou2023solitary,CHOI2022100406} for the analyses of NLS-type models with anisotropic diffraction whose numerical discretization yields nonlocal non-radial coupling.

\section{Constrained optimization and ground states}\label{optimization}
Define
\begin{equation}\label{constrained}
    I_\nu = \inf_{q \in l^2(\mathbb{Z}^d)} \{H[q]: M[q] = \nu\}.
\end{equation}
If $\nu = 0$, then $q = 0$, and therefore assume $\nu > 0$. A ground state at mass $\nu$ is defined as a minimizer of $I_\nu$. This section has two goals: (1) establishing sufficient conditions for the existence of ground states (\Cref{l2}), and (2) deriving the precise spatial decay of these ground states (\Cref{algebraic_decay}).

\subsection{Existence of ground states}
\begin{lemma}\label{l1}
$-\frac{\nu^{\frac{p+1}{2}}}{p+1} \leq I_\nu \leq 0$.
\end{lemma}
\begin{proof}
By $l^2 \hookrightarrow l^{p+1}$,
\begin{equation*}
H[q] \geq -\frac{\nu^{\frac{p+1}{2}}}{p+1} > - \infty.   \end{equation*}
Let $R>0$. Define $q_n = A := (\frac{\nu}{R^d})^{\frac{1}{2}}$ if $|n| \leq R$ and $q_n = 0$, otherwise; note that $M[q] \simeq \nu$. The second term in the Hamiltonian vanishes for large $R$ since
\begin{equation}\label{poten_est}
    \sum_n |q_n|^{p+1} \simeq \nu^{\frac{p+1}{2}} R^{-\frac{d(p-1)}{2}} \xrightarrow[R \rightarrow \infty]{} 0.
\end{equation}
In estimating $\langle Lq,q \rangle$, assume $|n| \leq R,\ |m| > R$ since otherwise, $|q_n - q_m|^2 = 0$.  Take $R>0,\ k_0 \in \mathbb{N}$ sufficiently large such that $k_0 \ll R$ and $(R + k_0)^d - R^d \simeq_d k_0 R^{d-1}$.

Consider $\mathcal{N} = \{|m^\prime| > R: \min\limits_{|n^\prime| \leq R} |m^\prime - n^\prime| < k_0 \}$. If $m \in \mathcal{N}$, then $J_{n-m}|q_n - q_m|^2 \leq \overline{J}\nu R^{-d}$. Since $|B(m;k_0) \cap B(0;R)| \leq |B(m;k_0)| \lesssim k_0^d$ and $|\mathcal{N}| \lesssim k_0 R^{d-1}$, we have
\begin{equation}\label{kinetic_est}
\sum_{m \in \mathcal{N},\ n \in B(m;k_0)} J_{n-m}|q_n - q_m|^2 \lesssim \nu \overline{J} k_0^{d+1}R^{-1}. 
\end{equation}
For $n \in B(0;R) \setminus B(m;k_0)$, set $j = n-m$. Then,
\begin{equation}\label{kinetic_est2}
    \sum_{m \in \mathcal{N},\ n \in B(0;R) \setminus B(m;k_0)} J_{n-m}|q_n - q_m|^2 \lesssim |\mathcal{N}| A^2 \sum_{|j|\geq k_0} J_{j} \lesssim k_0 \nu R^{-1} \sum_{|j|\geq k_0} J_{j}, 
\end{equation}
and therefore the RHS of \eqref{kinetic_est2} is bounded above by any $\epsilon > 0$ by taking $k_0, R$ sufficiently large, provided that $k_0 \ll R$.

Consider $m \in S := B(0;2R) \setminus (\mathcal{N} \cup B(0;R))$. A similar analysis yields
\begin{equation}\label{kinetic_est3}
    \sum_{m \in S,\ n \in B(0;R)} J_{n-m} |q_n - q_m|^2 \lesssim A^2 \sum_{|m| \leq 2R} \sum_{|j|\geq k_0} J_{j} \lesssim \nu \sum_{|j|\geq k_0} J_{j} = o(1),
\end{equation}
as $k_0 \rightarrow \infty$. Lastly let $|m| > 2R,\ |n| \leq R$. Then, the lower bound $|j| = |n-m| > R$ holds. Then the contribution to $\langle Lu,u\rangle$ is at most
\begin{equation}\label{kinetic_est4}
    \nu R^{-d} \sum_{|n| \leq R} \sum_{|m| > 2R} J_{n-m} \leq \nu R^{-d} \sum_{|n| \leq R} \sum_{|j| > R} J_{j} \lesssim \nu   \sum_{|j| > R} J_{j} = o(1),
\end{equation}
and hence the claim by \eqref{poten_est}, \eqref{kinetic_est}, \eqref{kinetic_est2}, \eqref{kinetic_est3}, \eqref{kinetic_est4}.
\end{proof}
The case when $I_\nu = 0$ can be characterized by the Gagliardo-Nirenberg-type inequality \eqref{gns}.
\begin{lemma}\label{l4}
    $I_\nu = 0$ if and only if
    \begin{equation}\label{gns}
        \| q \|_{l^{p+1}}^{p+1} \leq \frac{p+1}{2} \nu^{-\frac{p-1}{2}} \| q \|_{l^2}^{p-1} \langle Lq,q \rangle,\ \forall q \in l^2(\mathbb{Z}^d).    
    \end{equation}
\end{lemma}
\begin{proof}
    \eqref{gns} is equivalent to
    \begin{equation}\label{gns3}
        M[q] = \nu \Rightarrow H[q] \geq 0.
    \end{equation}
    If $I_\nu = 0$, then $H[q] \geq I_\nu = 0$ whenever $M[q] = \nu$. Conversely, \eqref{gns3} implies $I_\nu \geq 0$, which implies $I_\nu = 0$ by \Cref{l1}.
\end{proof}
If there exists $\nu > 0$ such that \eqref{gns} holds, then the inequality holds with $\nu$ replaced by $0 < \nu^\prime < \nu$. Motivated by this observation, define
    \begin{equation}\label{def:excite}
        \nu_0 := \sup \{\nu: \eqref{gns} \text{\ holds}\} \geq 0.
    \end{equation}
By continuity, $I_{\nu_0} = 0$. Equivalently, define 
\begin{equation}\label{gns5}
    K = \inf_{u \in l^2(\mathbb{Z}^d) \setminus \{0\} } \frac{\| u \|_{l^2}^{p-1} \langle Lu,u \rangle}{\| u \|_{l^{p+1}}^{p+1}} \geq 0.
\end{equation}
Comparing $K$ with the coefficients of \eqref{gns}, we have
\begin{equation}\label{excitation_threshold}
\nu_0 = \left(\frac{p+1}{2}\kappa K\right)^{\frac{2}{p-1}},    
\end{equation}
where $\kappa >0$ is the coupling coefficient if $L$ were replaced by $\kappa L$. In particular, \eqref{excitation_threshold} correctly predicts that there is no excitation threshold in the anti-continuum limit $\kappa \rightarrow 0$; a unique ground state of mass $\nu > 0$ is given by $q_* = \nu^{\frac{1}{2}}\delta$ where $\delta$ is the Kronecker delta supported at the origin (modulo translation). For $\nu$ sufficiently large, ground states are found from the strict negativity of $I_\nu$ as \cite[Theorem 2.1]{weinstein1999excitation}.
\begin{proposition}\label{l2}
    Let $I_\nu < 0$ and $\{q^{(k)}\} \subseteq l^2(\mathbb{Z}^d)$ be a minimizing sequence. Then, there exists a subsequence $\{q^{(n_k)}\},\ m_k \in \mathbb{Z}^d,\ \phi_k \in \mathbb{R}$, and $q_* \in l^2(\mathbb{Z}^d)$ with non-negative entries such that $e^{i \phi_k} q^{(n_k)}(\cdot - m_k) \xrightarrow[k \rightarrow \infty]{} q_*$ in $l^2(\mathbb{Z}^d)$, and hence $I_\nu = H[q_*]$. There exists $\omega(\nu) > 0$ such that
    \begin{equation}\label{main_model3}
        -w q_* = L q_* - q_*^p.
    \end{equation}
\end{proposition}
\begin{proof}
    Although a straightforward application of the method of concentration compactness is applied, we give a sketch of the proof to address a technical issue arising from LRI; for a detailed proof applied to DNLS, see \cite{Lions1984II} or the appendix of \cite{weinstein1999excitation}.
    
    There exists $\{q^{(n_k)}\}$ whose qualitative behavior is either of vanishing, dichotomy, or compactness (tight). We rule out the first two from which the existence of strong limit $q_* \in l^2(\mathbb{Z}^d)$ follows. Then the non-negativity follows by observing $\langle Lq,q \rangle \geq \langle L|q|,|q| \rangle$ since $|q_n - q_m| \geq ||q_n| - |q_m||$. Then, the Euler-Lagrange equation with $\omega(\nu)$ as the Lagrange multiplier yields \eqref{main_model3}.    

    By definition of minimizing sequence and $I_\nu$, there exists $\epsilon_k > 0$ such that $\epsilon_k \xrightarrow[k \rightarrow \infty]{} 0$ and, without re-labeling the subsequence indices, $I_\nu = H[q^{(k)}] - \epsilon_k$. By \Cref{l1},
    \begin{equation*}
        0 \geq I_\nu \geq -\frac{\| q^{(k)} \|_{l^{p+1}}^{p+1}}{p+1} - \epsilon_k. 
    \end{equation*}
    Vanishing subsequence implies $\| q^{(k)} \|_{l^{p+1}} \xrightarrow[k \rightarrow \infty]{} 0$ by the discrete analogue of \cite[Lemma I.1]{Lions1984II}, which implies $I_\nu = 0$.

    The dichotomy of concentration of mass cannot occur due to the strict subadditivity, $I_\nu < I_{\nu_1} + I_{\nu - \nu_1}$ for $\nu_1 \in (0,\nu)$, which follows straightforwardly from $I_{\theta \nu_1} < \theta I_{\nu_1}$ for $\theta > 1,\ \nu_1 > 0$. Indeed,    \begin{equation*}
        I_{\theta \nu_1} \leq \inf_{M[q] = \nu_1} \{H[\sqrt{\theta}q]\} = \theta \cdot \inf_{M[q] = \nu_1} \{\frac{1}{2}\langle Lq,q \rangle - \frac{\theta^{\frac{p-1}{2}}}{p+1}\| q \|_{l^{p+1}}^{p+1}\} < \theta I_{\nu_1},
    \end{equation*}
where the last inequality holds since no minimizing sequence has vanishing $l^{p+1}$ norm; otherwise, $I_{\nu_1} = 0$.

The claim is by contradiction. Let $\nu_1 \in (0,\nu)$. For any $\epsilon > 0$, suppose there exist $\{a^{(k)}\},\{b^{(k)}\} \subseteq l^2(\mathbb{Z}^d)$ with compact support and $R_k := \text{dist}(\text{supp}(a^{(k)}),\text{supp}(b^{(k)}))\xrightarrow[k\rightarrow \infty]{}\infty$ satisfying 
\begin{equation}\label{dichom}
    |M[a^{(k)}] - \nu_1| < \epsilon,\ |M[b^{(k)}] - (\nu - \nu_1)| < \epsilon,\ \| q^{(k)} - (a^{(k)} + b^{(k)})\|_{l^2} < \epsilon,
\end{equation}
for sufficiently large $k \in \mathbb{N}$. We claim $H[q^{(k)}] = H[a^{(k)}] + H[b^{(k)}] + o(1)$ as $k \rightarrow \infty,\ \epsilon \rightarrow 0$. It suffices to show $\langle La^{(k)},b^{(k)}\rangle \xrightarrow[k \rightarrow \infty]{}0$, which follows from
\begin{equation*}
\begin{split}
    \sum_{n} \sum_{m \neq n} J_{n-m} (a^{(k)}_n - a^{(k)}_m) b^{(k)}_n &= \sum_{n} \sum_{m \neq 0} J_{m}(a^{(k)}_n - a^{(k)}_{n+m}) b^{(k)}_n\\
    &= \sum_{m \neq 0} J_{m} \sum_{n} (a^{(k)}_n - a^{(k)}_{n+m})b^{(k)}_n\\
    &\lesssim \sum_{|m| \geq R_k} J_{m} (\nu_1 (\nu - \nu_1))^{\frac{1}{2}} \xrightarrow[k\rightarrow \infty]{}0.
\end{split}
\end{equation*}
Then, $I_\nu = H[a^{(k)}] + H[b^{(k)}] + o(1)$, which implies $I_\nu \geq I_{\nu_1} + I_{\nu - \nu_1}$ as $k \rightarrow \infty,\ \epsilon \rightarrow 0$.
\end{proof}
\begin{remark}
    An immediate consequence of any minimizing sequence being precompact modulo symmetries (translations and gauge) is that the set of ground states at any mass is orbitally stable. Spectral stability in the case of LRI is expected (see \cite{stefanov2023ground} for the spectral stability of DNLS), which we leave for further studies.
\end{remark}
\begin{remark}\label{r1}
    Taking the contrapositive of \eqref{gns}, if the inequality fails for some $q \in l^2(\mathbb{Z}^d)$ for some $\nu > 0$, then $I_\nu < 0$ by \Cref{l1}, and there exists a minimizer of $I_\nu$ by \Cref{l2}. Therefore for all $\nu > \nu_0$, we have $I_\nu < 0$, and conversely, for all $\nu \leq \nu_0$, $I_\nu = 0$ where $I_{\nu_0} = 0$ is by continuity.
\end{remark}

We give an equivalence to the existence of ground states at the threshold by the (inverse) discrete Gagliardo-Nirenberg attainment \eqref{gns5} and another application of the concentration compactness method. By \Cref{p1}, the threshold is positive precisely in the mass critical/supercritical regime.

\begin{lemma}\label{gn_equiv}
    Suppose $\nu_0 > 0$. There exists a minimizer $Q \in l^2(\mathbb{Z}^d)$ such that $M[Q] = \nu_0$ and $I_{\nu_0} = H[Q] = 0$ if and only if \eqref{gns5} is attained.
\end{lemma}

\begin{proof}
If \eqref{gns5} is attained at $\tilde{Q} \in l^2(\mathbb{Z}^d)$, we may assume $M[\tilde{Q}] = \nu_0$ by scaling. By \eqref{excitation_threshold}, it follows that $H[\tilde{Q}] = 0$. The converse statement follows similarly from $H[Q] = \frac{1}{2} \langle LQ,Q \rangle - \frac{\| Q \|_{l^{p+1}}^{p+1}}{p+1} = 0$.    
\end{proof}

\begin{proposition}
    Suppose $\nu_0 > 0$. Let $\{Q^{(n)}\} \subseteq l^2(\mathbb{Z}^d)$ be minimizers of $I_{\nu_n}$ where $\nu_n \rightarrow \nu_0 +$ as $n \rightarrow \infty$. Similarly, let $\{q^{(n)}\} \subseteq l^2(\mathbb{Z}^d)$ be a minimizing sequence for $I_{\nu_0} = 0$, i.e., $M[q^{(n)}] = \nu_0$ and $H[q^{(n)}] \rightarrow 0+$ as $n \rightarrow \infty$. If \eqref{gns5} is not attained, then both $\{Q^{(n)}\}$ and $\{q^{(n)}\}$ exhibit vanishing as characterized in \eqref{vanishing}. If \eqref{gns5} is attained, then $\{Q^{(n)}\}$ admits a strongly convergent subsequence, modulo gauge and translation symmetries, whose limit $Q_{*} \in l^2(\mathbb{Z}^d)$ is a ground state satisfying \eqref{main_model3} with $\omega > 0$. 
\end{proposition}

\begin{proof}
In general, dichotomy cannot happen for $q^{(n)}$ or $Q^{(n)}$; our argument below for $q^{(n)}$ generalizes to $Q^{(n)}$.  Consequently if \eqref{gns5} is not attained, then those sequences cannot be tight; otherwise, a strongly convergent subsequence can be extracted whose strong limit is a minimizer, which contradicts \Cref{gn_equiv}. For $q^{(n)}$, that the strong limit is a minimizer at $\nu_0$ satisfying the Euler-Lagrange equation follows from the Lagrange multiplier method as \Cref{l2}. On the other hand, $Q^{(n)}$ is a ground state at $\nu_n$, and the standard limit argument yields a ground state at $\nu_0$, if $\{Q^{(n)}\}$ were tight. Therefore for any $R > 0$ and without relabeling the subsequence, we obtain
\begin{equation}\label{vanishing}
    \sup_{m^\prime \in \mathbb{Z}^d} \sum_{m \in B(m^\prime,R)} |q^{(n)}_m|^2 \xrightarrow[n \rightarrow \infty]{} 0,
\end{equation}
and similarly for $Q^{(n)}$.

If the total mass splits into $\nu_1 \in (0,\nu_0)$ and $\nu_0 - \nu_1$, then we reason as \eqref{dichom} and decompose $q^{(n)} = a^{(n)} + b^{(n)} + o(1)$ as $n \rightarrow \infty$ where $o(1)$ is in $l^2(\mathbb{Z}^d)$. Then, $H[q^{(n)}] = H[a^{(n)}] + H[b^{(n)}] + o(1)$. We claim that $H[a^{(n)}] + H[b^{(n)}]$ is bounded away from the origin for all sufficiently large $n$, which forces a contradiction since $H[q^{(n)}] \rightarrow 0$.

For any $q \in l^2(\mathbb{Z}^d)$, \eqref{gns5} implies
\begin{equation}\label{coercive}
    H[q] \geq \left(1-\left(\frac{M[q]}{\nu_0}\right)^{\frac{p-1}{2}}\right)\frac{\langle Lq,q \rangle}{2}.
\end{equation}
By symmetry, let $q = a^{(n)}$. Since $M[a^{(n)}] \rightarrow \nu_1$ by hypothesis, the first factor of the RHS of \eqref{coercive} is bounded below from zero uniformly for all $n \gg 1$. To bound the second factor, there exist a sequence of centers $\{m_n\} \subseteq \mathbb{Z}^d$ and some $R,\delta>0$ such that $\sum\limits_{m \in B(m_n,R)}|a^{(n)}_{m}|^2 \geq \delta > 0$ for all $n \gg 1$. By considering the Dirichlet Rayleigh quotient, we obtain
\begin{equation*}
\langle L a^{(n)},a^{(n)} \rangle \gtrsim_R \sum\limits_{m \in B(m_n,R)}|a^{(n)}_{m}|^2 \geq \delta,    
\end{equation*}
for a fixed $R>0$, from which $H[a^{(n)}] + H[b^{(n)}] \gtrsim 1$ follows uniformly in $n \gg 1$. The same argument shows that dichotomy is impossible for $\{Q^{(n)}\}$.

Suppose \eqref{gns5} is attained at $\tilde{Q} \in l^2(\mathbb{Z}^d)$ where $M[\tilde{Q}] = \nu_0$ by scaling. We claim vanishing cannot happen for $Q^{(n)}$ and hence tightness. Then, a usual subsequence extraction yields a strong limit $Q_{*} \in l^2(\mathbb{Z}^d)$, modulo gauge and translation symmetries, that is a minimizer of the action at $\nu_0$.

Let $r_n = \frac{\nu_n}{\nu_0} > 1$ where the strict inequality may be assumed by taking a subsequence. Then, $M[r_n^{\frac{1}{2}}\tilde{Q}] = \nu_n$ and
\begin{equation}\label{ub}
    I_{\nu_n} \leq H[r_n^{\frac{1}{2}}\tilde{Q}] = -r_n (r_n^{\frac{p-1}{2}}-1) \frac{\langle L\tilde{Q},\tilde{Q} \rangle}{2}.
\end{equation}
Arguing as \eqref{coercive}, we have
\begin{equation}\label{lb}
    I_{\nu_n} = H[Q^{(n)}] \geq -(r_n^{\frac{p-1}{2}} - 1) \frac{\langle L Q^{(n)},Q^{(n)}\rangle}{2} ,
\end{equation}
Combining \eqref{ub}, \eqref{lb},
\begin{equation}\label{coercive2}
    \langle LQ^{(n)},LQ^{(n)}\rangle \geq r_n\langle L\tilde{Q},\tilde{Q} \rangle > \langle L\tilde{Q},\tilde{Q} \rangle > 0,
\end{equation}
for all $n \geq 1$. Vanishing implies $\| Q^{(n)}\|_{l^{p+1}} \xrightarrow[n \rightarrow \infty]{} 0$ by the discrete analogue of \cite[Lemma I.1]{Lions1984II}, which is incompatible with $I_{\nu_n} = H[Q^{(n)}] < 0$, or equivalently $0 \leq \langle L Q^{(n)}, Q^{(n)}\rangle < \frac{2 \| Q^{(n)}\|_{l^{p+1}}^{p+1}}{p+1}$, and the uniform lower bound \eqref{coercive2}.

\end{proof}

\begin{remark}
    If \eqref{gns5} is attained at $\tilde{Q} \in l^2(\mathbb{Z}^d)$ and $\{q^{(n)}\} \subseteq l^2(\mathbb{Z}^d)$ is a minimizing sequence for $I_{\nu_0}$, then either vanishing or tightness is possible. For the latter, let $q^{(n)} \equiv \left(\frac{\nu_0}{M[\tilde{Q}]}\right)^{\frac{1}{2}}\tilde{Q}$. For the former, let $q \in l^2(\mathbb{Z}^d)$ such that $M[q] = \nu_0$. For $x \in \mathbb{R}^d$, let $\phi(x) = (2\pi)^{-d} \int_{\mathbb{T}^d} \widehat{q}(\xi) e^{i x \cdot \xi} d\xi$ be the Shannon interpolation of $q$, and let $q^{(N)}_{m} = N^{-\frac{d}{2}} \phi\left(\frac{m}{N}\right)$. By the Plancherel Theorem for band-limited functions, $M[q^{(N)}] = M[q]$ for all $N \geq 1$. That $\{ q^{(N)} \}$ is a minimizing sequence follows from
    \begin{equation*}
        0 \leq H[q^{(N)}] \leq \langle Lq^{(N)},q^{(N)}\rangle = (2\pi)^{-d}  \int_{\mathbb{T}^d} \sigma(\xi) |\widehat{q^{(N)}}(\xi)|^2 d\xi = (2\pi)^{-d} \int_{[-\pi,\pi]^d} \sigma\left(\frac{\xi}{N}\right)|\widehat{\phi}(\xi)|^2 d\xi \xrightarrow[N \rightarrow \infty]{} 0,
    \end{equation*}
    where the first inequality is by \eqref{coercive} and the last limit is by the Dominated Convergence Theorem. Moreover, the sequence is vanishing since $\sup\limits_{n \in \mathbb{Z}^d}\| q^{(N)}\|_{l^2 (B(n,R))} \lesssim_d R^{\frac{d}{2}} \| q^{(N)} \|_{l^\infty } \leq R^{\frac{d}{2}} N^{-\frac{d}{2}} \nu_0^{\frac{1}{2}} \xrightarrow[N \rightarrow \infty]{}0$ for any $R>0$.   
\end{remark}

\subsection{Decay of ground states}\label{decay3}
It was claimed in \cite[Section 3]{gaididei1997effects} that when $d=1$ and $J_k = k^{-(1+\alpha)}$, the tail behavior of the localized states is exponential if $\alpha > 2$ and algebraic if $1 < \alpha < 2$. Soon after in \cite{flach1998breathers}, an analysis on a similar model (Klein-Gordan type with cubic potential and LRI) was reported with numerical evidence that the decay near the origin undergoes an exponential-to-algebraic crossover for $|n| > n_c$ for some transitional lattice site $n_c \in \mathbb{N}$. Here we provide a rigorous error analysis with the sharp order of decay. If $q \in l^2(\mathbb{Z}^d)$ is a ground state satisfying \eqref{main_model3}, then there exists $\omega > 0$ such that
\begin{equation}\label{ground_state}
    q = (\omega+L)^{-1}q^p = K \ast q^p
\end{equation}
where the resolvent kernel is given by
\begin{equation}\label{tail_decay}
    K_n = [\mathcal{F}^{-1} (\omega + \sigma(\xi))^{-1}]_n = (2\pi)^{-d} \int_{\mathbb{T}^d} \frac{e^{i n \cdot \xi}}{\omega + \sigma(\xi)}d\xi. 
\end{equation}
For $d=1, J_k = k^{-(1+\alpha)}$, \eqref{tail_decay} was approximated in \cite{gaididei1997effects} for $|n| \gg 1$ by
\begin{equation*}
    \frac{1}{2\pi}\int_{-\infty}^{\infty} \frac{\cos (n \xi)}{\omega + a\xi^2} d\xi = \frac{e^{-\sqrt{\frac{\omega}{a}}|n|}}{2\sqrt{a \omega}}
\end{equation*}
when $\alpha > 2$, for some $a>0$, using the properties of the polylogarithm, which yields an exponential decay of $q$. Note that $\sigma(\xi) \simeq |\xi|^2$ is valid for $|\xi| \ll 1$ by \Cref{l3}, but not for $|\xi| \gtrsim 1$, and therefore an extension of the integral from $\mathbb{T}^d$ to $\mathbb{R}^d$ should be dealt with care. The non-analytic branch term at $\xi = 0$ governs the asymptotic decay.

\begin{proposition}\label{algebraic_decay}
Let $q \in l^2(\mathbb{Z}^d)$ satisfy \eqref{main_model3} with $M[q] = \nu$. For any $\alpha \in (0,\infty)$, the decay estimate 
\begin{equation}\label{decay2}
    |q_n| \leq C(d,\alpha,p,\omega) \nu^{\frac{p}{2}} \langle n \rangle^{-(d+\alpha)},\ \forall n \in \mathbb{Z}^d 
\end{equation}
holds, and for $\alpha = \infty$, $q_n = O_{\beta}(|n|^{-\beta})$ for any $\beta > 0$. Furthermore, if $\epsilon > 0,\ \alpha \in (0,\infty)$, then
\begin{equation}
q_n \notin O(|n|^{-(d+\alpha+\epsilon)}).    
\end{equation}
\end{proposition}
\begin{proof}
By an explicit computation of $\mathcal{F}^{-1} [\sin^2 \left(\frac{m \cdot \xi}{2}\right)]_n$, we have
\begin{equation}\label{fourier}
    s_n := [\mathcal{F}^{-1} \sigma]_n = \sum_{m \neq 0} J_{m} \mathcal{F}^{-1} [\sin^2 \left(\frac{m \cdot \xi}{2}\right)]_n =
    \begin{cases}
            \frac{1}{2} \sum\limits_{m \neq 0} J_{m}
            & ,\text{ if } n = 0\\
            -\frac{1}{2} J_{n} & ,\text{ if } n \neq 0.
    \end{cases}
\end{equation}
Assume $\alpha \in (0,\infty)$. To obtain a decay estimate of $q$ from \eqref{ground_state}, an argument similar to the Picard iteration is implemented. Note that the decay of $K$ corresponds to the regularity of $(\omega + \sigma(\xi))^{-1}$, and therefore of $\sigma(\xi)$, which in turn corresponds to the decay of $s_n$; the symbols $(\omega + \sigma(\xi))^{-1}$ and $\sigma(\xi)$ are in the same H\"older space since $\sigma \geq 0$ and $\omega > 0$. In short, $K = O(|n|^{-(d+\alpha)})$.

By $l^2(\mathbb{Z}^d) \hookrightarrow l^\infty(\mathbb{Z}^d)$, for any $\epsilon > 0$, there exists $R>0$ such that $|q_n| \leq \epsilon$ whenever $|n| > R$. For $|n| > 2R$, the convolution in \eqref{ground_state} is split into $|m| \leq R$ and $|m| > R$ to obtain
\begin{equation}\label{ground_state2}
\begin{split}
    q_n &\leq \| q \|_{l^2}^p \sum_{|m| \leq R} K_{n-m} + \epsilon^{p-1}\sum_{|m| > R} K_{n-m} q_m\\
    &\lesssim \| q \|_{l^2}^p\sum_{|m| \leq R} \langle n-m \rangle^{-(d+\alpha)}+ \epsilon^{p-1}\sum_{|m| > R} K_{n-m} q_m =: S_1 + S_2.
\end{split}
\end{equation}
Since $|n-m| \geq \frac{|n|}{2}$, we have
\begin{equation*}
S_1 \leq C  \langle n \rangle^{-(d+\alpha)} \simeq \| q \|_{l^2}^p R^d \langle n \rangle^{-(d+\alpha)}.   
\end{equation*}
For $\delta > 0$, let $C_0(\delta) > 0$ be the best constant that satisfies
\begin{equation}\label{continuity}
    q_n \leq C_0(\delta + \langle n \rangle^{-(d+\alpha)}),\ \forall n \in \mathbb{Z}^d.
\end{equation}
Indeed $C_0$ exists since $\frac{q_n}{\delta + \langle n \rangle^{-(d+\alpha)}} \leq \frac{\nu^{1/2}}{\delta}$. Define $C_1 > 0$ by
\begin{equation*}
\sum\limits_{m} K_{n-m} \langle m \rangle^{-(d+\alpha)} \leq C_1 \langle n \rangle^{-(d+\alpha)},    
\end{equation*}
which can be shown by comparing the series with the corresponding integral. By \eqref{ground_state2}, \eqref{continuity},
\begin{equation*}
\begin{split}
    q_n &\leq C \langle n \rangle^{-(d+\alpha)} + C_0 \epsilon^{p-1} \sum_{|m| > R} K_{n-m} (\delta + \langle m \rangle^{-(d+\alpha)})\\
    &\leq (\epsilon^{p-1}\sum_{m} |K_{m}|) C_0\delta +  (C + C_0 (\epsilon^{p-1}C_1))\langle n \rangle^{-(d+\alpha)}\\
    &\leq (C + \frac{C_0}{2})(\delta + \langle n \rangle^{-(d+\alpha)}),
\end{split}
\end{equation*}
where $0 < \epsilon \ll 1$ was chosen at the last inequality, implying $C_0 \leq C + \frac{C_0}{2}$, and therefore $C_0 \leq 2C$. By \eqref{continuity} and $C>0$, the estimate \eqref{decay2} follows by taking $\delta \rightarrow 0$ in \eqref{continuity}. If $\alpha = \infty$, a similar analysis as above shows $O_{\beta}(|n|^{-\beta})$ decay for any $\beta > 0$.

Suppose there exists $\epsilon>0$ satisfying $d+\alpha + \epsilon \notin \mathbb{Z}$ and $\epsilon < (p-1)(d+\alpha)$ such that $q_n = O(|n|^{-(d+\alpha+\epsilon)})$. By \eqref{ground_state}, this implies $\mathcal{F}^{-1}[\sigma \widehat{q}]_{n} = (Lq)_{n} = (s \ast q)_{n} = O(|n|^{-(d+\alpha+\epsilon)})$. Let $R \gg 1$ such that $R^{-\frac{\epsilon}{2}}\sum\limits_{m} |s_m| \ll 1$ and $|q_{m_0}| > \frac{\| q \|_{l^\infty}}{2}$ for some $|m_0| < R$. Let $2R \leq |n| \leq 2^{-\frac{1}{d+\alpha}} R^{\frac{d+\alpha+\epsilon/2}{d+\alpha}}$ such that
\begin{equation*}
J_{n-m_0} \gtrsim \eta(|n-m_0|) \gtrsim |n|^{-(d+\alpha)},    
\end{equation*}
where the first inequality holds by the averaged lower bound in \Cref{def:kernel} and $|n-m_0| \geq R \gg 1$, and the second inequality holds by the limit property of $\eta$ and $|n-m_0| \leq \frac{3 |n|}{2}$. Then,
\begin{equation}
\begin{split}
    |n|^{-(d+\alpha+\epsilon)}\gtrsim|(s \ast q)_n| &\geq \left|\sum_{|m| \leq R} s_{n-m} q_m\right| - \left|\sum_{|m| > R} s_{n-m} q_m\right|\\
    &\gtrsim \sum_{|m| \leq R} J_{n-m} q_m - \sum_{|m| > R} \frac{|s_{n-m}|}{|m|^{d+\alpha+\epsilon}}\\
    &\gtrsim \| q \|_{l^\infty} |n|^{-(d+\alpha)} - (R^{-\frac{\epsilon}{2}} \sum_{m} |s_m|) R^{-(d+\alpha+\epsilon/2)} \gtrsim |n|^{-(d+\alpha)},
\end{split}
\end{equation}
where an explicit computation \eqref{fourier} and $q_n \geq 0$ were used in the third inequality.    
\end{proof}
\section{Excitation thresholds for LRI}\label{ET}
\subsection{Radial algebraic decay}\label{radial_algebraic}

The existence of excitation thresholds hinges upon the following lemma. Let $\langle z \rangle = (|z|^2 + 1)^{\frac{1}{2}}$. We introduce an effective interaction exponent $\alpha^\prime$ to describe the asymptotic behavior of the lattice Laplacian.

\begin{proposition}\label{l3}
Define  
\begin{equation}\label{prime}
    \alpha^\prime = 
    \begin{cases}
        \frac{\alpha}{2}
        & ,\text{ if } \alpha \in (0,2)\\
        1 & ,\text{ if } \alpha \in [2,\infty].
\end{cases}    
\end{equation}
Then,
\begin{equation}\label{equivalence}
    \sigma(\xi) \simeq_{d,\alpha} \delta(|\xi|) :=
        \begin{cases}
            |\xi|^{2\alpha^{\prime}}
            & ,\text{ if } \alpha \in (0,\infty] \setminus \{2\}\\
            |\xi|^2 \langle\log |\xi|\rangle & ,\text{ if } \alpha =2.
        \end{cases}
\end{equation}
\end{proposition}
For the following counting lemma, let $\#$ denote the cardinality of lattice points of interest.

\begin{lemma}\label{count}
Let $a \in (0,d^{-\frac{1}{2}})$. There exists $k_1(a,d) \in \mathbb{N}$ and $C(d) > 0$ such that
\begin{equation*}
    \sup_{\tilde{\omega} \in \mathbb{S}^{d-1}}\#\left\{m \in S_k: \frac{|m \cdot \tilde{\omega}|}{k} \leq a\right\} \leq C ak^{d-1},\ \forall k \geq k_1,
\end{equation*}
\end{lemma}

\begin{proof}

Since the statement is trivial for $d=1$, let $d \geq 2$. Fix $\tilde{\omega} \in \mathbb{S}^{d-1}$ and let $E = \{m \in S_k: \frac{|m \cdot \tilde{\omega}|}{k} \leq a\}$. Consider $\tilde{Q} = \bigcup\limits_{m \in E} (m + [-\frac{1}{2},\frac{1}{2}]^d)$. Define $T:\mathbb{R}^d \setminus \{0\} \rightarrow \mathbb{R}^d \setminus \{0\}$ given by $x \mapsto \frac{|x|_{1}}{|x|_{2}}x$. Then, $T$ maps $S_k$ into $\{x \in \mathbb{R}^d: |x|_2 = k\}$. In the polar coordinate $(r,\omega^\prime) \in \mathbb{R}_{>0} \times \mathbb{S}^{d-1}$, we have $T(x) = |\omega^\prime|_{1} r \omega^{\prime} = |\omega^\prime|_{1} x$ since $|\omega^{\prime}|_2 = 1$. Then, $J_T(x) = | \omega^\prime |_{1}^d \geq | \omega^\prime |_{2}^d = 1$, and
\begin{equation*}
\# E = \text{Vol}(\tilde{Q}) \leq \int_{\tilde{Q}} J_T(x) dx = \text{Vol}(T(\tilde{Q})).   
\end{equation*}

Let $x \in \tilde{Q}$. Then, there exists $m \in E$ such that $|x-m|_2 \leq \frac{\sqrt{d}}{2}$. Since $||x|_1 - |m|_1| \leq \sqrt{d}|x-m|_2$ and $|T(x)|_2 = |x|_1$, we have a radial bound: $|T(x)|_2 \in [k - \frac{d}{2},k + \frac{d}{2}]$. On the other hand, let $\omega^\prime_x = \frac{x}{|x|_2},\ \omega^\prime_m = \frac{m}{|m|_2}$, and $\angle(\omega^\prime_{x},\omega^\prime_{m}) \in [0,\pi]$ be the angle between two vectors. Since 
\begin{equation*}
|m|_2 \geq \frac{|m|_1}{\sqrt{d}} = \frac{k}{\sqrt{d}},\ |x|_2 \geq |m|_2 - |x-m|_2 \geq \frac{k}{\sqrt{d}} - \frac{\sqrt{d}}{2},
\end{equation*}
we have
\begin{equation*}
    \sin \angle(\omega^\prime_{x},\omega^\prime_{m}) \leq \frac{|x-m|_2}{\min (|x|_2,|m|_2)} \leq \frac{2 \sqrt{d}}{k},
\end{equation*}
for all $k \geq d$, and hence an angular bound: $\angle(\omega^\prime_{x},\omega^\prime_{m}) \lesssim_d \frac{1}{k}$ for all $k$ sufficiently large.

Since $m \in E$, we have $|\cos \angle(\omega^\prime_m,\tilde{\omega})| \leq a \sqrt{d} < 1$. Estimating $\cos^{-1}(a \sqrt{d})$, by the Taylor expansion for example, an angular bound $\angle(\omega^\prime_m,\tilde{\omega}) \in [\frac{\pi}{2}-2a\sqrt{d},\frac{\pi}{2}+2a\sqrt{d}]$ follows. Then,
\begin{equation*}
T(\tilde{Q}) \subseteq \left\{x \in \mathbb{R}^{d}: |x|_2 \in [k-\frac{d}{2},k+\frac{d}{2}] \text{ and } \angle(\omega^\prime_x,\tilde{\omega}) \in [\frac{\pi}{2}-(2a\sqrt{d} + \frac{2\sqrt{d}}{k}),\frac{\pi}{2}+(2a\sqrt{d}+\frac{2\sqrt{d}}{k})]\right\},    
\end{equation*}
and the volume of RHS (sector) is bounded above by $C_1 a k^{d-1} + C_2 k^{d-2}$ for some $C_1,C_2>0$. If $k \geq k_1(a,d)$, then the lower order term $k^{d-2}$ is bounded above $ak^{d-1}$ independent of $\tilde{\omega} \in \mathbb{S}^{d-1}$.   
\end{proof}

\begin{proof}[Proof of \Cref{l3}]
Assume $\alpha \in (0,\infty)$; the case $\alpha = \infty$ follows similarly. Fix $R \gg 1$ in \Cref{def:kernel} and let $|\xi| < R^{-1}$. If $|\xi| \geq R^{-1}$, then $\sigma(\xi) \simeq 1$ where the lower bound is by the nearest-neighbor degeneracy and the upper bound is by the compactness of $\mathbb{T}^d$, and hence $\sigma(\xi) \simeq \delta(|\xi|)$ away from the origin.

To show $\sigma(\xi) \lesssim \delta(|\xi|)$ on $|\xi| < R^{-1}$, note that
\begin{equation*}
\begin{split}
\sigma(\xi) &= \sum_{|m| \leq R} J_m \sin^2 \left(\frac{m\cdot \xi}{2}\right) + \sum_{R<|m| \leq |\xi|^{-1}} J_m \sin^2 \left(\frac{m\cdot \xi}{2}\right) + \sum_{|m| > |\xi|^{-1}} J_m \sin^2 \left(\frac{m\cdot \xi}{2}\right)\\
    &=: \mathscr{I}_1 + \mathscr{I}_2 + \mathscr{I}_3.    
\end{split}
\end{equation*}
By direct estimation using the upper bound $J_{m} \lesssim \eta(|m|)$,
\begin{equation*}
\mathscr{I}_1 + \mathscr{I}_2 + \mathscr{I}_3 \lesssim |\xi|^2 \sum_{|m| \leq R} J_{m} |m|^2 + |\xi|^2 \sum_{R < k \leq |\xi|^{-1}} k^{1-\alpha} + \sum_{k>|\xi|^{-1}} k^{-(1+\alpha)} \lesssim \delta(|\xi|).
\end{equation*}

To show the lower bound, we claim
\begin{equation}\label{count2}
    \mathbb{E}_{\mu_{k}} [(m \cdot \xi)^2] =  \mathbb{E}_{\mu_k} \left[\left(m \cdot \frac{\xi}{|\xi|}\right)^2\right] |\xi|^{2} \gtrsim k^2 |\xi|^2.
\end{equation}

Let $\tilde{\omega} = \frac{\xi}{|\xi|}$ and $0 < \delta \ll d^{-1}$. Define $S_\delta = \{ m \in S_k: \left(\frac{m \cdot \tilde{\omega}}{k}\right)^2 > \delta\}$. By \Cref{count},
\begin{equation*}
    \mu_{k} (S_\delta^c) \leq (\# S_\delta^{c}) \sup_{m \in S_k}\mu_{k}\{m\} \lesssim \sqrt{\delta},
\end{equation*}
where the last inequality follows from $\mu_{k}\{m\} \lesssim \frac{J_m}{k^{d-1} \eta(k)} \lesssim k^{-(d-1)}$. Then,
\begin{equation*}
    \mathbb{E}_{\mu_k} \left[\left(m \cdot \tilde{\omega}\right)^2\right] \geq k^2 \delta \mu_{k} (S_{\delta}) \geq C_1 k^2 \delta (1- C_2 \sqrt{\delta}),
\end{equation*}
for some dimensional constants $C_1,C_2>0$, which justifies the claim. Then for $R \leq k < |\xi|^{-1}$,
\begin{equation*}
    \sigma(\xi) \geq \sum_{R < |m| \leq |\xi|^{-1}} J_{m} \sin^2\left(\frac{m \cdot \xi}{2}\right) = \sum_{R < k \leq |\xi|^{-1}} \overline{J}_k \mathbb{E}_{\mu_k} \left[\sin^2\left(\frac{m \cdot \xi}{2}\right)\right] \gtrsim |\xi|^2\sum_{R < k \leq |\xi|^{-1}} \overline{J}_k k^2,
\end{equation*}
where the last inequality is by \eqref{count2}. By definition, $\overline{J}_k \gtrsim k^{d-1} \eta(k)$, and therefore,
\begin{equation*}
    |\xi|^2\sum_{R < k \leq |\xi|^{-1}} \overline{J}_k k^2 \gtrsim |\xi|^2 \sum_{R < k \leq |\xi|^{-1}} k^{1-\alpha} \gtrsim \delta(|\xi|).
\end{equation*}
\end{proof}

We now clarify how the excitation threshold depends on the nonlinearity. The next two propositions show a dichotomy: in the mass-subcritical regime there is no positive threshold, while at and above the mass-critical regime a strictly positive threshold appears. Weak nonlinear effects let energy disperse, whereas stronger interactions counteract dispersion and sustain localized states. On the other hand, when $\alpha = 0$, we give a partial result in $d=1$ where a logarithmic refinement yields positive thresholds for any power nonlinearity, thereby extending \Cref{p1}.

\begin{proposition}\label{p2}
If $p < 1 + \frac{4\alpha^\prime}{d}$, then $K = 0$. Consequently $\nu_0 = 0$, and $I_\nu$ has a minimizer for any $\nu > 0$.
\end{proposition}
The proposition is shown by an explicit construction of test functions similar to \cite[Proposition 4.1]{weinstein1999excitation}. An extension of this idea to LRI requires more technicality.
\begin{lemma}\label{quadratic_form}
Let $u_n = e^{-\rho |n|^{\gamma}}$ for $n \in \mathbb{Z},\ \rho>0,\ \gamma \in (0,2)$. Then as $\rho \rightarrow 0$, $\sum\limits_{n \in \mathbb{Z}} |u_n|^q \simeq (q \rho)^{-\frac{1}{\gamma}}$ holds for any $q \in (0,\infty)$. For $\rho \ll 1$,
\begin{equation}\label{quadratic_form2}
    \int_{-\pi}^{\pi} |\xi|^\beta |\widehat{u}(\xi)|^2 d\xi \simeq
    \begin{cases}
            \rho^{\frac{\beta-1}{\gamma}}
            & ,\text{ if } \beta \in (-1,1+2\gamma)\\
            \rho^2 |\log \rho| & ,\text{ if } \beta = 1+2\gamma\\
            \rho^2 & ,\text{ if } \beta > 1+2\gamma.
    \end{cases}
\end{equation}
\end{lemma}
\begin{proof}
Observe that $u_n$ is the fundamental solution of the fractional heat equation
\begin{equation}\label{frac_heat}
    \partial_\rho u = -(-\Delta)^{\frac{\gamma}{2}}u,\ (\xi,\rho) \in \mathbb{T}\times \mathbb{R}
\end{equation}
on the Fourier side where $\rho$ is the time variable. Let $p_\rho(x)$ be the fundamental solution of \eqref{frac_heat} posed on $\mathbb{R}$. 
From the semigroup theory of $\exp(-(-\Delta)^{\frac{\gamma}{2}})$, we have (see \cite[Section 2.6]{kwasnicki2017ten}, for example)
\begin{equation}\label{decay}
    p_\rho(x) \simeq \min(\rho^{-\frac{1}{\gamma}}, \frac{\rho}{|x|^{1+\gamma}}),\ x \in \mathbb{R}.
\end{equation}
On $\mathbb{T}$, the fundamental solution $\tilde{p}_\rho$ is periodized from $p_\rho$, i.e,
\begin{equation*}
    \tilde{p}_\rho(\xi) = \sum_{m \in \mathbb{Z}} p_\rho(\xi + 2m\pi) = \sum_{n\in\mathbb{Z}} e^{-\rho |n|^\gamma}\cos n\xi,
\end{equation*}
where the last equality is by the Poisson summation formula. By the relationship above, $\tilde{p}_\rho$ obeys \eqref{decay} for $|\xi| \leq \pi$, and therefore
\begin{equation*}
    \widehat{u}(\xi) = \tilde{p}_\rho(\xi) \simeq \frac{\rho}{\rho^{1+\frac{1}{\gamma}} + |\xi|^{1+\gamma}}.
\end{equation*}
Setting $\xi = 0$ and replacing $\rho$ by $q\rho$, the desired $l^q$-norm follows.

Let $I$ denote the integral in \eqref{quadratic_form2} where the domain is restricted to $[0,\pi]$ by the evenness of the integrand. Let $k_0 \geq 1$ be the least integer satisfying $\frac{\pi}{2^{k_0-1}} \leq \rho^{\frac{1}{\gamma}}$. Decomposing $I = I_1 + I_2 + I_3$ on $[0,\rho^{\frac{1}{\gamma}}] \cup [\rho^{\frac{1}{\gamma}},\frac{\pi}{2}] \cup [\frac{\pi}{2},\pi]$, respectively, where
\begin{align}\label{dyadic}
    I_1 &\simeq \rho^{-\frac{2}{\gamma}} \int_0^{\rho^{\frac{1}{\gamma}}} |\xi|^\beta d\xi \simeq \rho^{\frac{\beta-1}{\gamma}},\nonumber\\
    I_2 &\simeq \rho^2\sum_{k=1}^{k_0} \int_{\frac{\pi}{2^{k+1}}}^{\frac{\pi}{2^{k-1}}} \frac{2^{-\beta k}d\xi}{(\rho^{1+\frac{1}{\gamma}} + |\xi|^{1+\gamma})^2} \simeq \rho^2 \sum_{k=1}^{k_0} 2^{(1+2\gamma - \beta)k},\\
    I_3 &\simeq \rho^2 \int_{\frac{\pi}{2}}^{\pi} |\xi|^\beta d\xi \simeq \rho^2,\nonumber
\end{align}
\eqref{quadratic_form2} follows from the sign consideration of the geometric series in \eqref{dyadic}.
\end{proof}
\begin{proof}[Proof of \Cref{p2}]
Once $K=0$ is shown, the rest of the statement follows from \Cref{r1}. Let
\begin{equation*}
u_n = e^{-\rho |n|} = \prod\limits_{j=1}^d e^{-\rho |n_j|}.    
\end{equation*}
For $\rho \ll 1$,
\begin{equation*}
\| u \|_{l^2}^{p-1} \simeq  \rho^{-\frac{d(p-1)}{2}},\ \| u \|_{l^{p+1}}^{p+1} \simeq \rho^{-d}.    
\end{equation*}
Let $\alpha \in (0,\infty] \setminus \{2\}$. By \Cref{l3} and the Plancherel Theorem,
\begin{equation*}
    \langle Lu,u \rangle \simeq \sum_{i=1}^d \int_{-\pi}^{\pi} |\xi_i|^{2\alpha^\prime} |f(\xi_i)|^2 d\xi_i \cdot \prod_{j \neq i} \int_{-\pi}^{\pi} |f(\xi_j)|^2 d\xi_j,
\end{equation*}
where $f(\xi_j)$ is the discrete Fourier transform of $e^{-\rho |n_j|}$. By \Cref{quadratic_form},
\begin{equation}
    \int_{-\pi}^{\pi} |\xi_i|^{2\alpha^\prime} |f(\xi_i)|^2 d\xi_i = O(\rho^{2\alpha^\prime - 1}),\ \rho \rightarrow 0.
\end{equation}
Then by the condition on $p$,
\begin{equation*}
    \frac{\| u \|_{l^2}^{p-1}\langle Lu,u \rangle}{\| u \|_{l^{p+1}}^{p+1}} \lesssim \rho^{-\frac{d(p-1)}{2} + 2\alpha^\prime} \xrightarrow[\rho \rightarrow 0]{} 0.
\end{equation*}

If $\alpha = 2$, then $\sigma(\xi) \lesssim |\xi|^{2-\epsilon}$ for any $\epsilon > 0$. Choose $\epsilon > 0$ such that $p < 1 + \frac{2}{d}(2-\epsilon)$. Then a similar analysis as above yields $K = 0$.
\end{proof}

\begin{proposition}\label{p1}
If $p \geq 1 + \frac{4\alpha^\prime}{d}$, then $\nu_0 > 0$. If $\nu > \nu_0$, then $I_\nu < 0$ and has a minimizer. Conversely if $\nu < \nu_0$, then $I_\nu = 0$ and no minimizer exists.
\end{proposition}

\begin{proof}
Given $\nu_0 > 0$, \Cref{r1} implies the existence of minimizer for any $\nu > \nu_0$ and, conversely, $I_\nu = 0$ for any $\nu < \nu_0$. If $q_* \in l^2(\mathbb{Z}^d)$ were a minimizer of $I_\nu$ for $\nu < \nu_0$, then \eqref{gns} yields 
\begin{equation}\label{no_ground}
    \| q_* \|_{l^{p+1}}^{p+1}=\frac{p+1}{2} \langle Lq_*,q_*\rangle \leq \frac{p+1}{2} \nu_0 ^{-\frac{p-1}{2}} \nu^{\frac{p-1}{2}}\langle Lq_*,q_*\rangle,
\end{equation}
implying $\nu_0 \leq \nu$, a contradiction.

To show estimates of the form \eqref{gns} hold, let $\theta \in (0,1)$ satisfy $\frac{\theta \alpha^\prime}{d} = \frac{1}{2} - \frac{1}{p+1}$. Note by standard algebra that the condition on $\theta$ holds if and only if 
\begin{equation}\label{condition}
\begin{split}
\{\alpha \in [2,\infty],\ d \in \{1,2\}\} &\text{\ or\ } \{\alpha \in [2,\infty],\ d = 3,\ p < \frac{d+2}{d-2}\}\\
\text{or\ } \{\alpha \in [1,2),\ d = 1\} &\text{\ or\ } \{\alpha \in (0,2),\ d > \alpha,\ p < \frac{d+\alpha}{d-\alpha}\}.
\end{split}
\end{equation}
By \Cref{l3},
\begin{equation*}
    \| |\nabla_d|^{\alpha^\prime}u\|_{l^2}^2 \simeq \int |\xi|^{2\alpha^\prime}|\widehat{u}(\xi)|^2 d\xi \lesssim \int \sigma(\xi)|\widehat{u}(\xi)|^2 d\xi = \langle Lu,u \rangle,
\end{equation*}
where $|\nabla_d| = (-\Delta_d)^{\frac{1}{2}}$ and $\Delta_d$ is the graph Laplacian. Then by the discrete Gagliardo-Nirenberg inequality, we have
\begin{equation}\label{gns4}
    \| u \|_{l^{p+1}} \lesssim \| u \|_{l^2}^{1-\theta} \| |\nabla_d|^{\alpha^\prime}u\|_{l^2}^\theta \lesssim \| u \|_{l^2}^{1-\theta} \langle Lu,u \rangle^{\frac{\theta}{2}}.
\end{equation}
Let the parameters satisfy \eqref{condition}. Then \eqref{gns4} and $p \geq 1 + \frac{4\alpha^\prime}{d}$ imply
\begin{equation*}
    \| u \|_{l^{p+1}} \lesssim \| u \|_{l^{2}}^{\frac{p-1}{p+1}} \langle Lu,u \rangle^{\frac{1}{p+1}} \left(\frac{\langle Lu,u \rangle}{\| u \|_{l^{2}}^2}\right)^{\frac{d(p-1)-4\alpha^\prime}{4\alpha^\prime(p+1)}} \lesssim \| u \|_{l^{2}}^{\frac{p-1}{p+1}} \langle Lu,u \rangle^{\frac{1}{p+1}},
\end{equation*}
and hence $\nu_0 > 0$. If $p \gg 1$ does not satisfy \eqref{condition}, then we can bootstrap from the previous argument by choosing $p_1 < p$ satisfying \eqref{condition}, and hence
\begin{equation}\label{bootstrap}
    \sum_n |u_n|^{p_1+1} = \| u \|_{l^{p_1+1}}^{p_1+1} \lesssim \| u \|_{l^{2}}^{p_1-1} \langle Lu,u \rangle.
\end{equation}
Assume $\| u \|_{l^2} = 1$. By $l^2 \hookrightarrow l^\infty$, it follows that $\| u \|_{l^{p+1}}^{p+1} \lesssim \langle Lu,u \rangle$, and the general case follows from replacing $u$ by $\frac{u}{\| u \|_{l^2}}$.
\end{proof}

\begin{remark}
The quantity $\nu_{0}$ is the decisive mass threshold for ground-state existence, and this applies to \Cref{p1,p3,p4}. Hence in each setting, if $\nu>\nu_{0}$, then $I_\nu<0$ and a minimizer exists by \Cref{l2}, and if $\nu<\nu_{0}$, then $I_\nu=0$ and no minimizer exists.    
\end{remark}

\begin{corollary}\label{p4}
    Let $d = 1$. Suppose there exists $q > 1$ such that $C_1 \leq k \log^q (k) J_k \leq C_2$ for all $k \geq k_0$ for some $k_0 \in \mathbb{N}$ and $0< C_1 < C_2$. Then $\nu_0 > 0$ for any $p > 1$. 
\end{corollary}
\begin{proof}
We claim
\begin{equation}\label{symbol_log}
\sigma(\xi) \simeq_{q} \langle \log |\xi|\rangle^{-(q-1)}.
\end{equation}
Assuming \eqref{symbol_log}, let $\tilde{\alpha} \ll 1$ such that $1+ 2\tilde{\alpha} < p < \frac{1+\tilde{\alpha}}{1-\tilde{\alpha}}$ and $\frac{\theta \tilde{\alpha}}{2} = \frac{1}{2} - \frac{1}{p+1}$ for some $\theta \in (0,1)$; the general case when $p \geq \frac{1+\tilde{\alpha}}{1-\tilde{\alpha}}$ follows as the bootstrapping argument \eqref{bootstrap}. By \eqref{gns4}, \eqref{symbol_log}, it follows that 
\begin{equation*}
\| u \|_{l^{p+1}}^{p+1} \lesssim \| u \|_{l^{2}}^{p-1} \langle Lu,u \rangle    
\end{equation*}
for any $u \in l^2(\mathbb{Z})$, and hence $\nu_0 > 0$. It suffices to show \eqref{symbol_log} for $|\xi| < k_0^{-1}$. 

To obtain the upper bound, let $M \simeq (|\xi|^{\frac{2}{3}} |\log |\xi||^{\frac{q-1}{3}})^{-1}$. Then,
\begin{equation*}
\begin{split}
\sum_{|m| \leq M} J_{m}\sin^2\left(\frac{m  \xi}{2}\right) &\lesssim \overline{J}|\xi|^2 \sum_{k=1}^M k^{2} \lesssim |\xi|^2 M^{3} \simeq \langle\log |\xi|\rangle^{-(q-1)},\\
\sum_{|m| > M}  J_{m}\sin^2\left(\frac{m \xi}{2}\right) &\lesssim \sum_{k > M} \frac{1}{k \log^q k} \simeq \frac{1}{(\log M)^{q-1}} \lesssim \langle\log |\xi|\rangle^{-(q-1)}.
\end{split}
\end{equation*}
To obtain the lower bound, we show
\begin{equation*}
    \sigma(\xi) \geq \sum\limits_{|m| > |\xi|^{-1}} J_m \sin^2 \left(\frac{m \xi}{2}\right) \gtrsim \sum\limits_{k > |\xi|^{-1}} \frac{\sin^2 \left(\frac{k |\xi|}{2}\right)}{k \log^q k}.
\end{equation*}

Since $\sin^2 \left(\frac{k|\xi|}{2}\right) \geq \frac{1}{2}$ whenever $k \in E_j := [\frac{(2j+1/2)\pi}{|\xi|},\frac{(2j+3/2)\pi}{|\xi|}]$ for $j \in \mathbb{Z}$, it follows that
\begin{equation*}
\begin{split}
\sum\limits_{k > |\xi|^{-1}} \frac{\sin^2 \left(\frac{k |\xi|}{2}\right)}{k \log^q k} &\gtrsim \sum_{j=0}^\infty \sum_{k \in E_j} \frac{1}{k \log^q k}\\
&\simeq \sum_{j=0}^\infty \left(\log \frac{(2j+1)\pi}{|\xi|}\right)^{-(q-1)} - \left(\log \frac{(2j+\frac{3}{2})\pi}{|\xi|}\right)^{-(q-1)}\\
&\gtrsim \sum_{j=0}^\infty \frac{1}{(2j+3/2)\pi \log^q \left(\frac{(2j+3/2)\pi}{|\xi|}\right)} \gtrsim \sum_{k \gg |\xi|^{-1}} \frac{1}{k \log^q k} \simeq \langle \log |\xi|\rangle^{-(q-1)},
\end{split}
\end{equation*}
where the difference quotient was estimated below by the derivative of $(\log (\cdot))^{-(q-1)}$. 
\end{proof}
\subsection{Non-radial algebraic decay}\label{non-radial}

Although our analysis so far has focused on radially symmetric interaction kernels, the underlying variational framework applies more generally. In particular, it extends naturally to non-radial couplings, as illustrated in the example below. For $i=1,2$, let $J^{(i)}:\mathbb{Z} \rightarrow [0,\infty)$ be an interaction kernel of order $\alpha_i \in (0,\infty]$ and let $\alpha_1 \geq \alpha_2$ by symmetry. Define
\begin{equation*}
    J_{k_1,k_2} = J^{(1)}_{k_1} \delta_{k_2} + \delta_{k_1} J^{(2)}_{k_2}. 
\end{equation*}
Our motivation is to extend the ideas of \Cref{radial_algebraic} to characterize the ground states of the anisotropic 2D model given by \eqref{main_model} where
\begin{equation*}
\begin{split}
(L u)_{n,m} &= \sum_{(n^\prime,m^\prime) \neq (n,m)} J_{n-n^\prime,m-m^\prime}(u_{n,m}-u_{n^\prime,m^\prime})\\
&= \sum_{n^\prime \neq n} J^{(1)}_{n-n^\prime}(u_{n,m}-u_{n^\prime,m}) + \sum_{m^\prime \neq m} J^{(2)}_{m-m^\prime}(u_{n,m}-u_{n,m^\prime}) =: (L_1 u)_{n,m} + (L_2 u)_{n,m},     
\end{split}
\end{equation*}
and $L = \mathcal{F}^{-1} (\sigma(\xi_1,\xi_2)) \mathcal{F}$ where
\begin{equation}\label{symbol_2d}
    \sigma(\xi_1,\xi_2) = \sigma_1(\xi_1) + \sigma_2(\xi_2) := \sum\limits_{m \in \mathbb{Z}} J^{(1)}_{m} \sin^2 \left(\frac{m \xi_1}{2}\right) +  J^{(2)}_{m} \sin^2 \left(\frac{m\xi_2}{2}\right).
\end{equation}
\begin{proposition}
Let $\alpha_i^\prime$ be defined as \eqref{prime}. If $p < 1 + \frac{4\alpha_1^\prime\alpha_2^\prime}{\alpha_1^\prime + \alpha_2^\prime}$, then $I_\nu$ has a minimizer for any $\nu > 0$, or equivalently, $\nu_0 = 0$. 
\end{proposition}
\begin{proof}
As \Cref{p2}, it suffices to show the claim for $\alpha_i \neq 2$ for $i=1,2$. Let $\gamma = \frac{\alpha_2^\prime}{\alpha_1^\prime} \in (0,1]$ and define
\begin{equation*}
    u_{n,m} = v_n w_m := e^{-\rho |n|} e^{-\rho |m|^{\gamma}}.
\end{equation*}
A routine computation using \Cref{l3,quadratic_form} yields
\begin{equation*}
\begin{split}
\| u \|_{l^2}^{p-1} &\simeq \rho^{-\frac{p-1}{2}(1+\frac{1}{\gamma})},\ \| u \|_{l^{p+1}}^{p+1} \simeq \rho^{-(1+\frac{1}{\gamma})},\ \langle Lu,u \rangle \simeq \rho^{2\alpha_1^\prime - 1 - \frac{1}{\gamma}},    
\end{split}
\end{equation*}
and the claim follows as \Cref{p2}.
\end{proof}
To show $\nu_0 > 0$ for sufficiently large nonlinearity, we observe that the variational characterization of ground states extends straightforwardly from radial coupling to our consideration in this section. In particular, \Cref{l1,l4} and \Cref{l2} hold. We identify the range of parameters that ensures \eqref{gns} holds for all $l^2$-functions.
\begin{proposition}\label{p3}
    If $p \geq 1 + \frac{4\alpha_1^\prime\alpha_2^\prime}{\alpha_1^\prime + \alpha_2^\prime}$, then $\nu_0 > 0$. 
\end{proposition}
Since the dispersion relation depends on directions, the multiplier \eqref{symbol_2d} is not radial on $\mathbb{T}^2$, and to address this issue, the anisotropic modification of the Littlewood-Paley decomposition is reviewed. For $\alpha > 0$, define $D = \mathcal{F}^{-1}(\xi_1^2 + |\xi_2|^\alpha)^{1/2}\mathcal{F}$. Define an even function $\phi \in C^\infty_c((-2,2);[0,1])$ where $\phi(z) = 1$ for $z \in [-1,1]$. Let $\psi(z) := \phi(z)-\phi(2z)$. For $N \in 2^\mathbb{Z}$, let 
\begin{equation*}
    \psi_N(\xi_1,\xi_2) = \psi\left(\frac{\sqrt{\xi_1^2+|\xi_2|^\alpha}}{\pi N}\right);\ P_N := \mathcal{F}^{-1}\psi_N \mathcal{F}.
\end{equation*}
On lattices, the dyadic decomposition is restricted to $N \leq 1$. The mapping properties of $D$ and $P_N$ were studied in $\mathbb{R}^2$ in \cite[Lemma 3.1, 3.3]{CHOI2022100406}. These results can be transferred to $\mathbb{Z}^2$ with ease.
\begin{lemma}\label{Bernstein}
Let $1 \leq q_1 \leq q_2 \leq \infty,\ N \leq 1$, and $\alpha > 0$. Then,
\begin{equation*}
\| P_N f \|_{l^{q_2}(\mathbb{Z}^2)} \lesssim N^{(1+\frac{2}{\alpha})(\frac{1}{q_1}-\frac{1}{q_2})} \| P_N f \|_{l^{q_1}(\mathbb{Z}^2)}.    
\end{equation*}
\end{lemma}
\begin{lemma}\label{Bernstein2}
Let $s \in \mathbb{R}$, $r \in [1,\infty],\ N \leq 1$. Then,
\begin{equation*}
    \| D^{s} P_N f \|_{l^r(\mathbb{Z}^2)} \simeq N^s \| P_N f \|_{l^r(\mathbb{Z}^2)}.
\end{equation*}
\end{lemma}

As \Cref{p1}, the key idea is to show the Gagliardo-Nirenberg inequality involving the operator $D$.

\begin{proof}[Proof of \Cref{p3}]
Adopting the method in \cite[Proposition 4]{Younghun} along with \Cref{Bernstein,Bernstein2}, the Gagliardo-Nirenberg inequality
\begin{equation}\label{gns2}
    \| u \|_{l^{p+1}} \lesssim \| u \|_{l^2}^{1-\theta} \| D^s u \|_{l^2}^\theta
\end{equation}
can be shown for $s > 0,\ \theta \in (0,1),\ \alpha > 0$ satisfying $s\theta = (1+\frac{2}{\alpha})(\frac{1}{2} -  \frac{1}{p+1})$. Let $\alpha = \frac{2\alpha_2^\prime}{\alpha_1^\prime},\ s = \alpha_1^\prime$. Then,
\begin{equation}\label{quadratic_form3}
    \| D^s u\|_{l^2}^2 \simeq \iint (|\xi_1|^{2s} + |\xi_2|^{\alpha s}) |\widehat{u}(\xi)|^2 d\xi_1 d\xi_2 \lesssim \langle Lu,u \rangle
\end{equation}
by \Cref{l3}. By \eqref{gns2}, \eqref{quadratic_form3},
\begin{equation*}
    \| u \|_{l^{p+1}}^{p+1} \lesssim  \| u \|_{l^2}^{p-1} \langle Lu,u \rangle \left(\frac{\langle Lu,u \rangle}{\| u\|_{l^2}^2}\right)^{\frac{\alpha_1^\prime+\alpha_2^\prime}{4\alpha_1^\prime \alpha_2^\prime}\left(p-\left(1 + \frac{4\alpha_1^\prime\alpha_2^\prime}{\alpha_1^\prime + \alpha_2^\prime}\right)\right)} \lesssim \| u \|_{l^2}^{p-1} \langle Lu,u \rangle,
\end{equation*}
and hence $\nu_0 > 0$. If $p > 1$ does not satisfy the scaling condition for \eqref{gns2}, then the claim follows from the bootstrapping argument \eqref{bootstrap} by applying \eqref{gns2} to $p_1 \ll p$. 
\end{proof}

\subsection{Estimation of excitation thresholds}\label{ET2}
An estimation of \eqref{excitation_threshold} is given at the anti-continuum regime. An explicit computation of the derivatives in $\kappa$ is needed to estimate thresholds.

\begin{proposition}\label{onsite}
Let $J$ be as \Cref{def:kernel} and $\nu > 0$. Then, there exists $\kappa_0 >0$ and a smooth map $[0,\kappa_0) \rightarrow l^2(\mathbb{Z}^d)$ given by $\kappa \mapsto q(\kappa)$ such that $\{q(\kappa)\}$ defines a family of unique ground states satisfying $-\omega(\kappa) q = \kappa L q - q^p$ and $M[q(\kappa)] = \nu$ where
\begin{equation}\label{omega}
    \omega(\kappa;\nu) = \nu^{\frac{p-1}{2}} -\kappa \sum_{m \neq 0} J_{m} - \kappa^2\frac{(p-3)\nu^{-\frac{p-1}{2}}}{2} \sum_{m \neq 0} J_{m}^2 + O(\kappa^3).
\end{equation}
Furthermore $q(\kappa)$ is a unique minimizer of \eqref{gns5} for any $\kappa \in (0,\kappa_0)$, and therefore
\begin{equation}\label{excitation_threshold2}
    \nu_0 = \left(\frac{p+1}{2} \overline{J} \kappa\right)^{\frac{2}{p-1}} + o(\kappa^{\frac{2}{p-1}}),\ \kappa \rightarrow 0.
\end{equation}
\end{proposition}

The expansion \eqref{omega} shows that in the anti-continuum limit ($\kappa \rightarrow 0$), the excitation threshold depends on the leading-order contribution from $\overline{J}$, which governs interaction strength between neighboring sites. The absence of nonzero coupling at $\kappa = 0$ results in a vanishing threshold, implying that every excitation is localized when interactions are strictly local.

\begin{proof}
Let $\delta_{n_0}$ be the Kronecker delta supported at $n_0 \in \mathbb{Z}^d$. Let $\Phi: X \times \mathbb{R}_{\kappa} \rightarrow X:=l^2(\mathbb{Z}^d) \times \mathbb{R}_{\omega}$ where $(q,\omega,\kappa) \mapsto ((\omega + \kappa L)q - q^p, \| q \|_{l^2}^2 - \nu)=: (\Phi_1,\Phi_2)$. By \Cref{l2}, a ground state has non-negative entries and hence it suffices to consider $l^2(\mathbb{Z}^d)$ as a Hilbert space over the reals. If $\kappa = 0$, the nonlinear eigenvalue equation reduces to $\omega q = q^p$, and therefore $q_n \in \{\omega^{\frac{1}{p-1}},0\}$. Since any ground state is necessarily supported at a single site, let $z = (q_*,\omega_*,\kappa_*) := (\nu^{\frac{1}{2}} \delta_{n_0},
\nu^{\frac{p-1}{2}},0)$ where we may assume $n_0 = 0$ by the translation symmetry. It can be shown that the Fr\'{e}chet derivative in $X$ at $z$, $D\Phi(z)$, is invertible, and therefore the Implicit Function Theorem (IFT) yields the desired (unique) smooth map on $[0,\kappa_0)$. Taking the $\kappa$-derivatives of $\Phi(q(\kappa),\omega(\kappa)) = 0$ gives 
\begin{equation*}
\begin{split}
    (\omega + \kappa L - pq^{p-1}) \dot{q} + (\dot{\omega} + L)q &= 0;\ \langle q, \dot{q} \rangle =0,\\
    (\omega + \kappa L -pq^{p-1}) \ddot{q} + (2\dot{\omega} + 2L - p(p-1)[q^{p-2}\dot{q}])\dot{q} + \ddot{\omega} q &= 0;\ \langle \dot{q},\dot{q} \rangle + \langle q, \ddot{q}\rangle = 0.
\end{split}
\end{equation*}
An explicit evaluation at $\kappa = 0$ yields
\begin{equation*}
\begin{split}
    \dot{q}(0)_n &= \nu^{\frac{2-p}{2}} J_{n}(1-\delta);\ \dot{\omega}(0) = -\sum_{m \neq 0} J_{m}\\
    \ddot{q}(0)_n &= 
    \begin{cases}
            -\nu^{\frac{3-2p}{2}} \sum\limits_{m \neq 0} J_{m}^2
            & ,\ n = 0\\
            2 \nu^{\frac{3-2p}{2}}\sum\limits_{m \notin \{0,n\}} J_{m} J_{n-m} & ,\ n \neq 0
    \end{cases};
    \qquad \ddot{w}(0) = (3-p) \nu^{-\frac{p-1}{2}} \sum_{m \neq 0} J_{m}^2,
\end{split}
\end{equation*}
and therefore
\begin{equation}\label{taylor}
    q(\kappa)_n = 
    \begin{cases}
            \nu^{\frac{1}{2}}-\frac{\nu^{\frac{3-2p}{2}}}{2}\kappa^2 \sum\limits_{m \neq 0} J_{m}^2 + o(\kappa^2)
            & ,\ n = 0\\
            \nu^{\frac{2-p}{2}}\kappa J_{n} + o(\kappa)& ,\ n \neq 0,
    \end{cases}
\end{equation}
and likewise for $\omega$, which shows \eqref{omega}. Note that the error terms of $q(\kappa)$ can be uniformly bounded in $n$ since the convolution terms of $\ddot{q}(0)$ can be uniformly bounded by the Young's inequality. Better approximations could be obtained by taking the higher order Taylor expansions.

To show $q(\kappa)$ is a unique minimizer for $\kappa > 0$, assume $M[v] = \nu$; otherwise consider $\tilde{v} = \frac{\nu^{1/2}}{\| v \|_{l^2}} v$. Since $H[v] \geq H[q(\kappa)]$, we have
\begin{equation}\label{minimizer}
    \frac{\langle \kappa L v,v \rangle}{\| v \|_{l^{p+1}}^{p+1}} \geq \frac{2H[q(\kappa)]}{\| v \|_{l^{p+1}}^{p+1}} + \frac{2}{p+1}.
\end{equation}
Then the RHS of \eqref{minimizer} $> \frac{\langle \kappa L q(\kappa),q(\kappa) \rangle}{\| q(\kappa) \|_{l^{p+1}}^{p+1}}$, and hence the claim, if and only if $\| v \|_{l^{p+1}} < \| q(\kappa) \|_{l^{p+1}}$ for any $\kappa \in (0,\kappa_0)$ by shrinking $\kappa_0$ if necessary.

By contradiction, suppose there exists $\{\kappa_j\} \subseteq (0,\kappa_0)$ such that $\kappa_j \xrightarrow[j\rightarrow\infty]{} 0$ and $\| v \|_{l^{p+1}} \geq \| q(\kappa_j) \|_{l^{p+1}}$. By \eqref{taylor}, $\| q(\kappa_j) \|_{l^{p+1}} \xrightarrow[j\rightarrow\infty]{} \| q(0) \|_{l^{p+1}} = \nu^{\frac{1}{2}}$ strictly from the right, and therefore for any $j$ sufficiently large, we have
\begin{equation*}
\nu^{\frac{1}{2}} < \| q(\kappa_j)\|_{l^{p+1}} \leq \| v \|_{l^{p+1}} \leq \| v \|_{l^{2}} = \nu^{\frac{1}{2}}. 
\end{equation*}
\end{proof}
\begin{remark}
    In \Cref{fig}, the dynamics under \eqref{main_model} with the coupling $\kappa = 0.1,\ d=1,\ J_{n} = |n|^{-(1+\alpha)},\ p=3$ is generated with the initial condition $q(\kappa) = \nu^{\frac{1}{2}} \delta + \frac{\nu^{-\frac{1}{2}}\kappa}{|n|^{1+\alpha}} (1-\delta)$, i.e., the first-order approximation of $q(\kappa)$ in \eqref{taylor}. The vanishing Dirichlet boundary condition, $u_{n} = 0$ for $|n| > 50$, is imposed. Observe in \eqref{excitation_threshold2} that if the zeroth-order approximation of $q(\kappa)$, i.e., $q(\kappa) \approx \nu^{\frac{1}{2}}\delta$, were used, then $\nu_0 \approx 4\kappa \zeta(1+\alpha)$. Hence if $\alpha \ll 1$, then $\nu_0 \gg 1$ and no ground state can form by \Cref{p1}, which can be seen in the left plot. Similarly the right plot suggests the solution to decay for $\nu \ll \nu_0$ for fixed $\alpha$.
\end{remark}
\begin{remark}
    Variational approaches do not ensure uniqueness in general. Uniqueness in \Cref{onsite} is a consequence of IFT, and therefore cannot be generalized to an arbitrary $\kappa$. For PDEs, uniqueness of ground states to nonlocal equations have been established for certain models  \cite{frank2016uniqueness}. It is of interest whether similar results could be obtained for nonlocal models in lattices or networks. 
\end{remark}

\vspace*{-5.0cm}
\begin{figure}[h]
\hspace*{-2.5cm}
\begin{subfigure}[h]{0.75\linewidth}
\includegraphics[width=\linewidth]{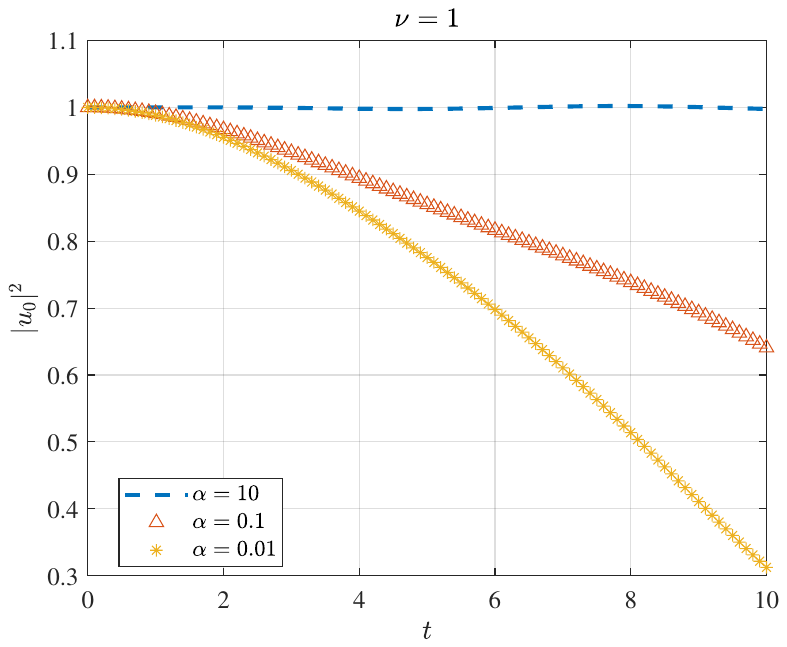}
%\caption{$\alpha = 0.5$}
\end{subfigure}
\hspace*{-4.0cm}
\begin{subfigure}[h]{0.75\linewidth}
\includegraphics[width=\linewidth]{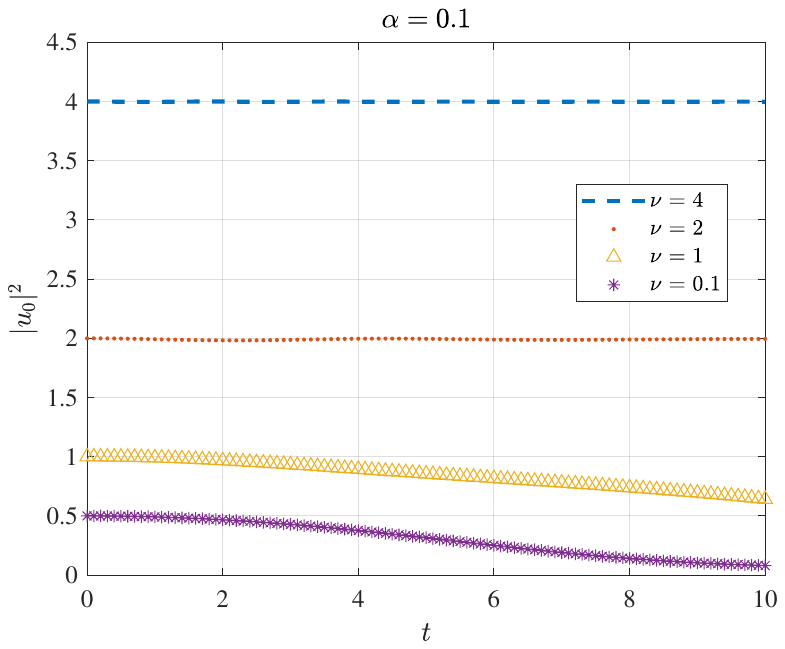}
%\caption{$\alpha=50$}
\end{subfigure}
\vspace*{-5.0cm}
\caption{Peak intensity vs. time for varying $\alpha$ (left) and $\nu$ (right) with $\kappa = 0.1$. The left plot shows a transition from localization to decay as $\alpha$ increases. The right plot confirms that for $\nu < \nu_0$, solutions exhibit global dispersive decay, while for $\nu > \nu_0$, localized breather-like states persist, validating the threshold predictions of \Cref{p1}.}
\label{fig}
\end{figure}

\section{Applications to the 1D algebraic lattice}\label{applications}
For concreteness, let $d=1,\ J_n = |n|^{-(1+\alpha)}$ for $\alpha < \infty$. For initial data of small mass below the energy threshold, ground states do not exist by \Cref{p1}. In fact, it was shown in \cite{stefanov2005asymptotic} that the solutions to DNLS decay to zero in $l^r$ for $r > 2$. Analogously, nonlinear decay estimates and global scattering are shown in this section by developing linear dispersive estimates and realizing the nonlinear evolution as perturbation. On the other hand, we provide numerical evidence that localized initial data with sufficiently large mass evolve into a localized breather-like component and an algebraically-decaying radiation.

In \Cref{decay_transition}, the left panel illustrates the decomposition into a breather-like core and a dispersive component. The right panel shows the algebraic decay rate of radiation, which follows $|n|^{-(1+\alpha)}$, with fitted slopes $-2.499$ (for $\alpha = 1.5$) and $-1.501$ (for $\alpha = 0.5$), closely matching theoretical predictions.

\begin{figure}[h!]
\vspace*{-5.0cm}
\hspace*{-2.0cm}
\begin{subfigure}[h]{0.75\linewidth}
\includegraphics[width=\linewidth]{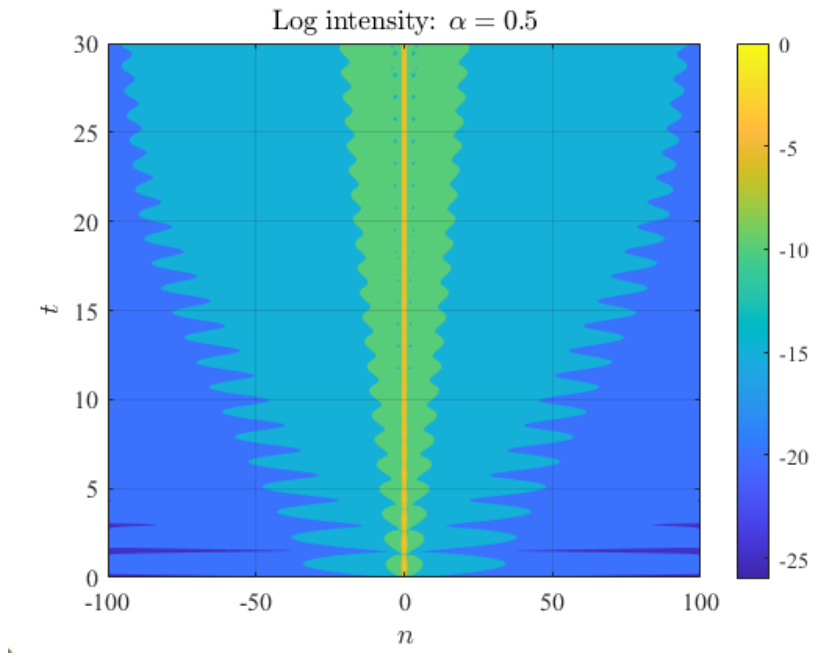}
%\caption{$\alpha = 0.5$}
\end{subfigure}
\hspace*{-4.0cm}
\begin{subfigure}[h]{0.75\linewidth}
\includegraphics[width=\linewidth]{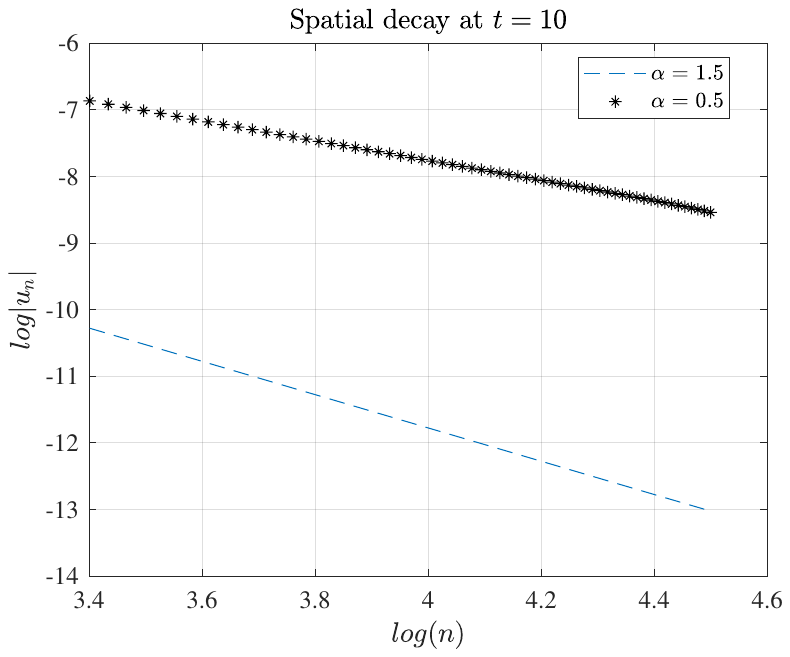}
%\caption{$\alpha=50$}
\end{subfigure}
\vspace*{-5.0cm}
\caption{The evolution of a localized initial condition $u(0) = 5 \delta_0$ under $\kappa = 0.1,\ p=3$, and the Dirichlet boundary condition $u_{n} = 0,\ |n| > 100$. Left: $\log |u_n(t)|$ on $t \in [0,30]$ for $\alpha = 0.5$. Right: for $\alpha = 1.5$, the linear fit is $y = -2.499 x - 1.779$, and for $\alpha = 0.5$, $y = -1.501 x - 1.756$. The left plot is a visual illustration of the decomposition into a localized state and radiation. Note that the intensity is not stationary since $u(0)$ is not an exact solution to the nonlinear eigenvalue problem \eqref{main_model3}. The decay of radiation, whose order is $d+\alpha$, is in excellent agreement with \Cref{algebraic_decay}.}
\label{decay_transition}
\end{figure}

To show linear dispersive estimates, the basics of the polylogarithm functions are revised. The linear flow of \eqref{main_model} with $(u_n(0))_n = (\phi_n)_n \in l^2$ is given by
\begin{equation}\label{kt}
    u(t) = U(t) \phi = K_{t} \ast \phi;\qquad (K_t)_n := \frac{1}{2\pi} \int_{-\pi}^{\pi} e^{i (n \xi-\sigma(\xi)t)}d\xi.
\end{equation}
By scaling, let
\begin{equation*}
    \sigma(\xi) = \sum_{k=1}^\infty \frac{\sin^2\left(\frac{k\xi}{2}\right)}{k^{1+\alpha}} = \frac{1}{2}\zeta(1+\alpha) - \frac{1}{4}(\mathrm{Li}_{1+\alpha}(e^{i\xi}) + \mathrm{Li}_{1+\alpha}(e^{-i\xi})),
\end{equation*}
where the polylogarithm $\mathrm{Li}_s(z)$ is a holomorphic function in $s,z \in \mathbb{C}$; see \cite{NIST:DLMF} and references therein for more details regarding this special function. The radius of convergence (in $z$) is 1, inside which
\begin{equation*}
    \mathrm{Li}_s(z) := \sum_{k=1}^\infty \frac{z^k}{k^s}.
\end{equation*}
Outside the radius of convergence, there exists an analytic continuation in $\mathbb{C} \setminus [1,\infty)$ where for $\Re s>0$, the polylogarithm can be identified with the Bose-Einstein integral
\begin{equation}\label{bose_einstein}
    \mathrm{Li}_s(z) = \frac{1}{\Gamma(s)} \int_0^\infty \frac{t^{s-1}}{e^t/z - 1} dt.
\end{equation}
An important property of polylogarithms pertinent to our interest is
\begin{equation}\label{derivative}
    \frac{d}{d\xi} \mathrm{Li}_s(e^{\pm i \xi}) = \pm i \mathrm{Li}_{s-1}(e^{\pm i \xi}),\ s \in \mathbb{C},\ \xi \notin 2\pi \mathbb{Z}.
\end{equation}
Although $\sigma \in C(\mathbb{T})$ by the summability of $J$, it is not analytic at the origin. Hence the derivatives of $\sigma$ are locally identified with those of the polylogarithms, by uniqueness of analytic continuation, which furnishes the following dispersive estimates.

\begin{proposition}[Linear dispersive decay]\label{p5}
    For any $\alpha \in (0,\infty)$, there exists $C = C(\alpha) > 0$ such that
    \begin{equation}\label{dispersive_estimate}
        \| U(t) \phi \|_{l^\infty} \leq C \langle t \rangle^{-\rho}\| \phi \|_{l^1},
    \end{equation}
    where $\rho = \rho_+:=\frac{1}{3}$ if $\alpha \in (1,\infty)$ and $\rho = \rho_{-}:=\frac{1}{2}$ if $\alpha \in (0,1]$.    
\end{proposition}

\begin{remark}
Our proof, based on the Van der Corput Lemma, which was also used in \cite{stefanov2005asymptotic}, splits into three cases: i) $\alpha \in (0,1)$; ii) $\alpha = 1$; iii) $\alpha \in (1,\infty)$. The case $\alpha = 1$ reduces to the well-known dilogarithm. The other cases are based on detailed estimates using the Bose-Einstein integral. For the sake of the overall flow of presentation, the proof is given in \Cref{appendix}.    
\end{remark}

\begin{figure}[h]
\vspace*{-5.0cm}
\begin{center}
\includegraphics[scale=0.6]{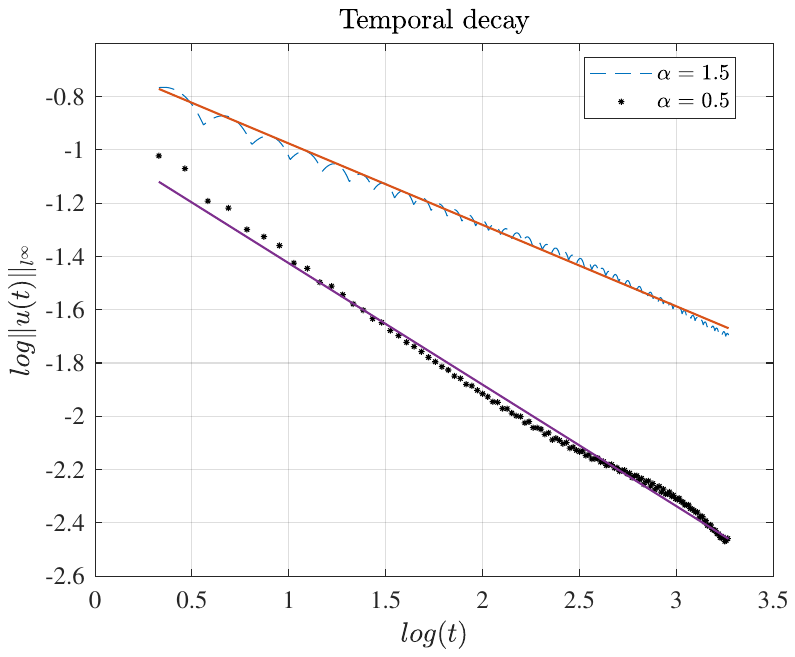}    
\end{center}
\vspace*{-5.5cm}
\caption{Peak intensity versus time with $(\kappa,p) = (1,3)$. For $\alpha = 1.5$, the linear fit is $y = -0.30555 x - 0.67053$, and for $\alpha = 0.5$, $y = -0.45583 x - 0.96974$ computed in Matlab. The vanishing Dirichlet boundary condition, $u_{n} = 0,\ |n| > 100$ is used. As predicted in \Cref{p5}, the rates of dispersive decay $\rho_{\pm}$ is illustrated by the slopes of the log-log plot. This numerically verifies \Cref{nonlinear_decay,scattering2} where small data disperse over the lattice without self-collapse.}
\label{decay_time}
\end{figure}

Let $2 \leq q,r \leq \infty$. Say $(q,r)$ is $\rho_{+}$-admissible if $\frac{3}{q} + \frac{1}{r} \leq \frac{1}{2}$ and $\rho_{-}$-admissible if $\frac{2}{q} + \frac{1}{r} \leq \frac{1}{2}$. By applying \cite[Theorem 1.2]{keel1998endpoint}, we obtain
\begin{corollary}[Strichartz estimates]
For any $\rho_{+}$-admissible pairs (or $\rho_{-}$-admissible pairs) $(q,r),(\tilde{q},\tilde{r})$, we have
\begin{equation}\label{strichartz}
\begin{split}
\| U(t) \phi \|_{L^q_t l^r} &\lesssim \| \phi \|_{l^2},\\
\left|\left| \int_{0}^{t} U(t-\tau) F(\tau) d\tau \right|\right|_{L^q_t l^r} &\lesssim \| F\|_{L^{\tilde{q}^\prime}_t l^{\tilde{r}^\prime}}.
\end{split}
\end{equation}
\end{corollary}

Assuming small data, the linear dispersive estimates extend to the nonlinear regime under certain technical hypotheses. The conditions below are most likely not sharp.
\begin{proposition}[Nonlinear dispersive decay]\label{nonlinear_decay}
Suppose that
\begin{equation*}
p > 
\begin{cases}
            3
            & ,\ \alpha \in (0,1]\\
            \max(4,1+4\alpha^\prime) & ,\ \alpha \in (1,\infty)
\end{cases}
;\qquad q \in
\begin{cases}
            (\frac{2(p+1)}{p-1},p+1]
            & ,\ \alpha \in (0,1]\\
            (\frac{2(p+1)}{p-2},p+1] & ,\ \alpha \in (1,\infty).
    \end{cases}    
\end{equation*}
Then,
\begin{equation}\label{nonlinear_estimate}
    \| u(t) \|_{l^q} \leq C \langle t \rangle^{-(\rho_{\pm})\frac{(q-2)}{q}} \| \phi \|_{l^{q^\prime}},
\end{equation}
whenever $\| \phi \|_{l^{q^\prime}} \ll_{\alpha,p,q} 1$ where $C(\alpha,p,q)>0$ depends on the constants of the Strichartz estimates \eqref{strichartz}.
\end{proposition}
\begin{proof}
The proof is in the spirit of \cite[Theorem 7]{stefanov2005asymptotic} where an a priori bound is obtained to run a fixed point argument. We only give a sketch of proof for $\alpha \in (0,1]$ to illustrate the source of our technical hypotheses.

Let $X = \{(u_n(t)): \| u(t) \|_{l^q} \lesssim \langle t \rangle^{-(\rho_{-})\frac{(q-2)}{q}} \| \phi \|_{l^{q^\prime}}\}$ and 
\begin{equation}\label{duhamel}
\Gamma u := U(t) \phi + i \int_0^t U(t-t^\prime) \left(|u(t^\prime)|^{p-1}u(t^\prime)\right)dt^\prime.    
\end{equation}
To show the existence of a fixed point in $X$, and therefore $u = \Gamma u$, observe that the interpolation of $U(t)$ for $l^1 \rightarrow l^\infty$ \eqref{dispersive_estimate} and $l^2 \rightarrow l^2$ (unitary action) yields \eqref{nonlinear_estimate} with $u(t)$ replaced by $U(t)\phi$ for all $q \in [2,\infty]$. The condition $q \leq p+1$ ensures $\| u(t^\prime) \|_{l^{pq^\prime}} \leq \| u(t^\prime) \|_{l^{q}}$, and the lower bound on $q$ yields the estimate $\int_0^t \langle t-t^\prime\rangle^{-\frac{q-2}{2q}} \langle t^\prime \rangle^{-\frac{p(q-2)}{2q}} \lesssim \langle t \rangle^{-\frac{q-2}{2q}}$. Taking $\| \phi \|_{l^{q^\prime}} \ll_{\alpha,p,q} 1$, an a priori estimate
\begin{equation*}
\begin{split}
\| \Gamma u \|_{l^q} &\lesssim \langle t \rangle^{-(\rho_{-})\frac{(q-2)}{q}} \| \phi \|_{l^{q^\prime}} + \int_0^t \langle t-t^\prime \rangle^{-(\rho_{-})\frac{(q-2)}{q}} \| u(t^\prime)\|_{l^{pq^\prime}}^p dt^\prime\\
&\lesssim \langle t \rangle^{-(\rho_{-})\frac{(q-2)}{q}} ( \| \phi \|_{l^{q^\prime}} + \| \phi \|_{l^{q^\prime}}^p) \lesssim \langle t \rangle^{-(\rho_{-})\frac{(q-2)}{q}} \| \phi \|_{l^{q^\prime}},
\end{split}
\end{equation*}
holds and the standard contraction argument yields the desired nonlinear decay.
\end{proof}

A stronger assumption on the power of nonlinearity yields a well-defined wave operator and global scattering for small data.
\begin{proposition}[Scattering]\label{scattering2}
Suppose $p \geq 5$ if $\alpha \in (0,1]$ and $p \geq 7$ if $\alpha \in (1,\infty)$. There exists $\delta > 0$ such that whenever $\| \phi \|_{l^2} < \delta$, there exists a unique asymptotic state $u_{\pm} \in l^2(\mathbb{Z})$ with 
$\| u_{\pm} \|_{l^2} = \| \phi \|_{l^2}$ such that
\begin{equation}\label{scattering}
 u(t) - U(t)u_{\pm} \xrightarrow[t\rightarrow \pm\infty]{} 0 \text{ in } l^2(\mathbb{Z}). 
\end{equation}
Furthermore the map $\phi \mapsto u_{\pm}$ defines homeomorphism on the $\delta$-neighborhood of the origin. 
\end{proposition}
\begin{proof}
The proof is standard in the literature, having developed the linear dispersive estimates. However we illustrate the main ideas for $\alpha \in (0,1]$ for the sake of completeness. Assume $t \in [0,\infty)$ without loss of generality.

A fixed point argument using $\Gamma: X \rightarrow X$ where $\Gamma$ is by \eqref{duhamel} and $X = C(\mathbb{R};l^2) \cap L_t^{\frac{4p}{3}}l^{\frac{2p}{p-3}}$ yields a global solution $u \in X$ with $\| \phi \|_{l^2} \ll 1$ where the implicit constant depends on $\alpha,p$ and the Strichartz constants. Since
\begin{equation*}
    \left|\left| \int_{t_1}^{t_2} U(-t^\prime) \left(|u(t^\prime)|^{p-1}u(t^\prime)\right)dt^\prime\right|\right|_{l^2} \lesssim \| u\|_{L_{t \in [t_1,t_2]}^{\frac{4p}{3}}l^p}^p \leq \| u\|_{L_{t \in [t_1,t_2]}^{\frac{4p}{3}}l^{\frac{2p}{p-3}}}^p \xrightarrow[t_1,t_2 \rightarrow \infty]{} 0,
\end{equation*}
there exists $u_+ = \lim\limits_{t \rightarrow \infty} U(-t)u(t)\in l^2$ such that \eqref{scattering} holds by reasoning as above. Injectivity follows from the uniqueness of limit.

To show surjectivity, let $u_+ \in l^2(\mathbb{Z})$ and define
\begin{equation*}
\tilde{\Gamma} u(t) := U(t) u_{+} - i \int_t^{\infty} U(t-t^\prime)\left(|u(t^\prime)|^{p-1}u(t^\prime)\right)dt^\prime,
\end{equation*}
and let $T \gg 1$ such that $\| U(t) u_+\|_{L_{t \in [T,\infty)}^{\frac{4p}{3}}l^{\frac{2p}{p-3}}} \ll 1$. Another fixed point argument using $\tilde{\Gamma}$ yields $u \in C([T,\infty);l^2) \cap L_t^{\frac{4p}{3}}l^{\frac{2p}{p-3}}$ such that its asymptotic state is $u_+$. Then flow $u(T)$ backwards in time to obtain $\phi = u(0)$. Continuity follows from the contractivity of $\Gamma,\tilde{\Gamma}$ and the Strichartz estimates.
\end{proof}

\section{Conclusion}

We analyzed excitation thresholds in multi-dimensional DNLS lattices with nonlocal interactions, deriving explicit conditions on how dimension $d$, nonlinearity $p$, and the nonlocal exponent $\alpha$ determine breather formation. Our results reveal that long-range interactions fundamentally alter excitation thresholds and dispersive decay, with a sharp transition at $\alpha = 1$. For $\alpha = 0$, the standard framework breaks down due to the lack of summability, requiring a novel logarithmic correction. This result highlights the robustness of excitation thresholds even in the extreme long-range regime. These findings have implications for nonlinear wave systems such as optical waveguides, Bose-Einstein condensates, and photonic lattices, where tuning $\alpha$ could enable precise control over localization and energy transport. Future work includes analyzing the spectral and dynamical stability of breather solutions, extending results to anisotropic and disordered lattices, and refining numerical methods for higher-dimensional settings.

\section*{\fontsize{12}{20}\selectfont Acknowledgement}
This work was partially supported by the United States Military Academy via the Faculty Research Funds. The author thanks A Aceves and anonymous reviewers for constructive comments. 
\appendix
\section{Appendix}\label{appendix}
\begin{proof}[Proof of \Cref{p5}]
Our proof is based on the Van der Corput Lemma (VCL) applied to $K_t$ in \eqref{kt}, and hence the higher-order derivatives of $\sigma$ are estimated from below. By unitarity, we have $\| U(t) \phi \|_{l^\infty} \leq \| \phi \|_{l^1}$, and therefore assume $|t| \geq 1$. By \eqref{derivative},
\begin{equation*}
    \sigma^{\prime\prime}(\xi) = \frac{1}{4}\left(\mathrm{Li}_{\alpha-1}(e^{i\xi}) + \mathrm{Li}_{\alpha-1}(e^{-i\xi})\right);\  \sigma^{\prime\prime\prime}(\xi) = \frac{i}{4}\left(\mathrm{Li}_{\alpha-2}(e^{i\xi}) - \mathrm{Li}_{\alpha-2}(e^{-i\xi})\right).    
\end{equation*}
For special cases, $\alpha \in \{1,2\}$, we show by direct computation. Let $\alpha = 1$. Since $\mathrm{Li}_{0}(z) = \frac{z}{1-z}$ for $z \neq 1$, it follows that $|\sigma^{\prime\prime}(\xi)| = \frac{1}{4}$ for $\xi \neq 0$. By VCL,
\begin{equation*}
    2\pi |(K_t)_n| = \left|\int_{-\pi}^{\pi} e^{it (\frac{n}{t}\xi - \sigma(\xi)t)}d\xi\right| \leq 2\delta + \left|\int_{[-\pi,\pi]\setminus (-\delta,\delta)} e^{it (\frac{n}{t}\xi - \sigma(\xi)t)}d\xi\right| \leq 2\delta + C |t|^{-\frac{1}{2}} \leq 2C |t|^{-\frac{1}{2}},
\end{equation*}
where $\delta \leq \frac{C|t|^{-1/2}}{2}$. For $\alpha = 2$ and $z \in \mathbb{C} \setminus [1,\infty)$, it can be shown that $\mathrm{Li}_1(z) = - \log(1-z)$ and
\begin{equation*}
\sigma^{\prime\prime}(\xi) = -\frac{1}{4} \left(\log(1-e^{i\xi})+\log(1-e^{-i\xi})\right),\ \sigma^{\prime\prime\prime}(\xi) = -\frac{1}{4} \cot(\frac{\xi}{2}).    
\end{equation*}
Near the origin, $|\sigma^{\prime\prime\prime}(\xi)| \gtrsim 1$, and near $\xi = \pm \pi$, we have $\sigma^{\prime\prime}(\xi) = -\frac{\log 2}{2} + \frac{|k \mp \pi|^2}{16} + O(|k \mp \pi|^3)$ by the Taylor expansion, and hence the decay of $O(|t|^{-\frac{1}{3}})$ by VCL.

For $\alpha > 2$, the integral representation \eqref{bose_einstein} is applied. If $|\xi| \in (\frac{\pi}{2},\pi]$, then
\begin{equation*}
|\mathrm{Li}_{\alpha-1}(e^{i\xi})+\mathrm{Li}_{\alpha-1}(e^{-i\xi})| = \frac{2}{\Gamma(\alpha-1)}\int_0^\infty \frac{t^{\alpha-2}(1-e^t \cos \xi)}{e^{2t}-2 e^t \cos \xi + 1}dt\geq \frac{1}{2\Gamma(\alpha-1)}\int_0^\infty t^{\alpha-2}e^{-2t} dt = 2^{-\alpha},
\end{equation*}
since $e^{2t}-2 e^t \cos \xi + 1 \leq 4 e^{2t}$ and $1-e^t \cos \xi \geq 1$.

Let $\xi_0 \ll 1$. For $|\xi| \in (0,\xi_0]$, let $t_0 = -\log \cos \xi > 0$. For $t \in [0,t_0]$, the first-order Taylor expansion yields
\begin{equation*}
\frac{1-e^t \cos \xi}{e^{2t}-2 e^t \cos \xi + 1} \leq \frac{(-\cos\xi) t + 1-\cos\xi}{2 - 2\cos\xi}
\end{equation*}
to obtain
\begin{align}\label{nogada}
&\frac{\Gamma(\alpha-1)}{2}|\mathrm{Li}_{\alpha-1}(e^{i\xi})+\mathrm{Li}_{\alpha-1}(e^{-i\xi})| \geq \int_{t_0}^\infty \frac{t^{\alpha-2}(e^t \cos \xi - 1)}{e^{2t}-2 e^t \cos \xi + 1}dt - \int_0^{t_0} \frac{t^{\alpha-2}( 1-e^t \cos \xi)}{e^{2t}-2 e^t \cos \xi + 1}dt\nonumber\\
&\geq \int_{t_0}^\infty \frac{t^{\alpha-2}(e^t \cos \xi - 1)}{2e^{2t}}dt - \int_{0}^{t_0} \frac{t^{\alpha-2}((-\cos\xi) t + 1-\cos\xi)}{2 - 2\cos\xi}dt\nonumber\\
&=\frac{1}{2}(-\log \cos \xi)^{\alpha-1}\left(\frac{-\cos \xi \cdot\log \cos \xi}{\alpha (1-\cos\xi)} - \frac{1}{\alpha-1} + \int_1^\infty t^{\alpha-2}\left((\cos \xi)^{t+1} - (\cos\xi)^{2t}\right)dt\right),
\end{align}
where \eqref{nogada} from its previous step could be assisted with softwares such as Maple or Mathematica. For $t < \epsilon_0 \xi^{-2}$ for some $0 < \epsilon_0 \ll 1$, we have
\begin{equation*}
    \int_1^\infty t^{\alpha-2}\left((\cos \xi)^{t+1} - (\cos\xi)^{2t}\right)dt \geq \frac{\xi^2}{4} \int_1^{\epsilon_0 \xi^{-2}} t^{\alpha-2} (t-1)dt \gtrsim \frac{\xi^{2-2\alpha}}{\alpha (\alpha-1)},
\end{equation*}
and therefore
\begin{equation*}
    \eqref{nogada} \gtrsim (-\log \cos \xi)^{\alpha-1} \xi^{2-2\alpha} \gtrsim 1,
\end{equation*}
where it can be shown by direct computation that $\frac{-\cos \xi \cdot\log \cos \xi}{\alpha (1-\cos\xi)} - \frac{1}{\alpha-1} \geq 0$ for any $0 \leq \xi < \frac{\pi}{2}$.

For $|\xi| \in (\xi_0,\frac{\pi}{2}]$,
\begin{equation*}
    |\mathrm{Li}_{\alpha-2}(e^{i\xi})-\mathrm{Li}_{\alpha-2}(e^{-i\xi})| = \frac{2 |\sin \xi|}{\Gamma(\alpha-2)}\int_0^\infty \frac{t^{\alpha-3} e^t}{e^{2t}-2 e^t \cos \xi + 1}dt \geq |\sin \xi_0| > 0,
\end{equation*}
by $e^{2t}-2 e^t \cos \xi + 1 \leq 2e^{2t}$. Hence $\max\limits_{\xi \neq 0}(|\sigma^{\prime\prime}(\xi)|,|\sigma^{\prime\prime\prime}(\xi)|) \gtrsim_{\alpha} 1$.

For $\alpha \in (1,2)$, although \eqref{bose_einstein} does not yield a convergent integral for $\sigma^{\prime\prime\prime}(\xi)$, the derivative relation \eqref{derivative} gives
\begin{equation*}
    4\sigma^{\prime\prime\prime}(\xi) = \frac{d}{d\xi}(\mathrm{Li}_{\alpha-1}(e^{i\xi}) + \mathrm{Li}_{\alpha-1}(e^{-i\xi})) = \frac{1}{\Gamma(\alpha-1)} \int_0^\infty t^{\alpha-2} \frac{d}{d\xi} \frac{2(e^t \cos\xi - 1)}{e^{2t}-2 e^t \cos \xi + 1}dt,
\end{equation*}
as long as $\xi \neq 0$. Similarly for $\alpha \in (0,1)$, $|\sigma^{\prime\prime}(\xi)|$ can be estimated from below by $\frac{d}{d\xi}(\mathrm{Li}_{\alpha}(e^{i\xi}) - \mathrm{Li}_{\alpha}(e^{-i\xi}))$. Then estimating the integrals as above ($\alpha > 2$), the proposition follows.
\end{proof}

%\clearpage
\bibliographystyle{ieeetr}
\small
\bibliography{ref}

\end{document}